\documentclass[twoside,leqno]{article}

\usepackage[letterpaper]{geometry}
\usepackage[utf8]{inputenc}

\usepackage[utf8]{inputenc}
\usepackage{amsmath}
\usepackage{amssymb}
\usepackage{cancel}
\usepackage{enumerate}
\usepackage{hyperref}
\usepackage{mathtools}
\usepackage{stmaryrd}
\usepackage{graphics}
\usepackage{ amssymb }

\usepackage{ltexpprt}

\usepackage{float}

\usepackage{todonotes}

\usepackage[linesnumbered,boxed,algo2e]{algorithm2e}

 \newtheorem{definition}{Definition}

 \newtheorem{claim}{Claim}

\newcommand{\s}{\sigma}
\newcommand{\T}{\mathcal{T}}
\renewcommand{\l}{\ell}
\newcommand{\hC}[1]{C_{#1} \cup Y_{#1} \cup \{z\}}

\newcommand{\C}{\mathcal{C}}

\renewcommand{\L}{\mathcal{L}}
\newcommand{\Y}{\mathcal{Y}}
\renewcommand{\S}{\mathcal{S}}
\newcommand{\N}{\mathcal{N}}

\begin{document}

\title{\Large Recognizing $k$-leaf powers in polynomial time, for constant $k$}
\author{Manuel Lafond\footnote{Université de Sherbrooke}}

\date{}

\maketitle

\vspace{10mm}

\begin{abstract}
    A graph $G$ is a $k$-leaf power if there exists a tree $T$ whose leaf set is $V(G)$, and such that $uv \in E(G)$
if and only if the distance between $u$ and $v$ in $T$ is at most $k$.  The graph classes of $k$-leaf powers have several applications in computational
biology, but recognizing them has remained a challenging algorithmic problem for the past two decades.  The best known result is that $6$-leaf powers can be recognized in polynomial time.
In this paper, we present an algorithm that decides whether a graph $G$ is a $k$-leaf power in time $O(n^{f(k)})$ for some function 
$f$ that depends only on $k$ (but has the growth rate of a power tower function).

Our techniques are based on the fact that either a $k$-leaf power has a corresponding tree of low maximum degree, in which case finding it is easy, 
or every corresponding tree has large maximum degree.  In the latter case, large degree vertices in the tree imply that $G$ has redundant substructures which can be pruned from the graph. In addition to solving a longstanding open problem, we hope that the structural results presented in this work can lead to further results on $k$-leaf powers.
\end{abstract}


\section{Introduction}

In computational biology, it is commonplace to use dissimilarity information between species to reconstruct a phylogenetic tree, in which the leaves are the species and the internal nodes represent common ancestors. 
Sequence distances can be used for this task, but are known to be unreliable~\cite{felsenstein2004inferring,philippe2011resolving}.  In 2002, Nishimura et al.~\cite{nishimura2002graph} proposed that each pair of species should simply be considered as either \emph{close} or \emph{far}.  There should then be a threshold $k$ such that in the phylogeny, close species are at distance at most $k$ and far species at distance more than $k$.  The $k$-leaf power problem arises when we model species as the vertices of a graph $G$ in which edges represent closeness.

More specifically, we say that a graph $G$ is a \emph{$k$-leaf power} if there exists a tree $T$ such that the set of leaves of $T$ is $V(G)$, and such that $uv \in E(G)$ if and only if  $dist_T(u, v) \leq k$, where $dist_T(u, v)$ is the distance between $u$ and $v$ in $T$ and $u, v$ are distinct.
The tree $T$ is called a \emph{$k$-leaf root} of $G$.
The $G$ graph is a \emph{leaf power} if it is a $k$-leaf power for some $k$.

Since their introduction, these graph classes have attracted the attention of both algorithm designers and graph theoreticians.  
Two important problems have remained open for the last two decades.  The first is to obtain a precise graph-theoretical characterization of $k$-leaf powers in terms of $k$.  A fundamental question asks whether, for all $k$,  $k$-leaf powers can be characterized as chordal graphs that forbid a finite set of induced subgraphs.  This is known to be true for $k = 2, 3, 4$, but unknown for higher $k$~\cite{brandstadt2006structure,brandstadt2008structure}. 

The second open problem is whether one can decide in polynomial time whether a graph $G$ is a $k$-leaf power, where $k$ could be fixed or given.  This has been a longstanding problem even in the case $k \in O(1)$.  Polynomial-time recognition is possible for $k \leq 6$~\cite{chang2006linear,ducoffe20194}, and the technical feats required to solve the case $k = 6$ show that extending these results is far from trivial.
In this work, we tackle the latter question and show that polynomial-time recognition is indeed possible for any constant $k$.

\subsection{Related work}

It is well-known that all $k$-leaf powers are chordal, and they are also strongly chordal (see e.g.~\cite{brandstadt2010rooted}). 
The $2$-leaf powers are collections of disjoint cliques and 
the $3$-leaf powers are exactly the chordal graphs that are bull, dart and gem-free~\cite{brandstadt2006structure,rautenbach2006some}.  This can be used to recognize them in linear time.
For $k = 4$, a characterization of twin-free $4$-leaf powers in terms of chordality and a small set of forbidden induced subgraphs is established, again leading to a linear time algorithm~\cite{brandstadt2008structure}.
For $k \geq 5$, a characterization still escapes us, except for distance-hereditary $5$-leaf powers~\cite{brandstadt2009forbidden}.  
It is known that all $k$-leaf powers are also $(k+2)$-leaf powers, but that they are not all $(k+1)$-leaf powers~\cite{brandstadt2008k}.

Chang and Ko~\cite{chang20073} have developed a linear time recognition algorithm for $5$-leaf powers, using a reduction from a similar problem known as the $3$-Steiner root, in which members of $V(G)$ can also be in internal nodes in the desired tree. 
Recently, Ducoffe showed that $6$-leaf powers were polynomial-time recognizable~\cite{ducoffe20194}, this time using a reduction from the $4$-Steiner root problem and an elaborate dynamic programming approach. The case $k = 6$ is the farthest that could be achieved so far.  
Let us mention that $k$-leaf powers have bounded clique-width~\cite{gurski2009nlc}, and that in~\cite{dom2004error,dom2005extending}, the problem of editing a graph to a $k$-leaf power is studied.

A recent result of Eppstein and Havvaei~\cite{eppstein2020parameterized} states that recognizing $k$-leaf powers is fixed parameter tractable (FPT) in $k + \delta(G)$, where $\delta(G)$ is the degeneracy of the graph.  Note that since a $k$-leaf power $G$ is chordal, it is known that $tw(G) = \omega(G) - 1 \leq \delta(G) \leq tw(G)$, where $tw(G)$ and $\omega(G)$ are the treewidth and clique number, respectively, showing that the problem is also FPT in $tw(G) + k$.
We also mention that in~\cite{chen2003computing}, Chen et al. show that if we require each node of the $k$-leaf root to have degree between $3$ and $d$ (except leaves), then finding such a $k$-leaf root, if any, is FPT in $k + d$.  

Recognizing the larger class of leaf powers is also a challenging open problem.  Subclasses of strongly chordal graphs have been shown to be leaf powers~\cite{kennedy2006strictly,brandstadt2008ptolemaic,nevries2016towards}, but not all strongly chordal graphs are leaf powers~\cite{brandstadt2010rooted,lafond2017strongly,jaffke2019mim}.
Other variants and generalizations of $k$-leaf powers have also been proposed~\cite{brandstadt2007k,calamoneri2016pairwise}.

\subsection{Our contributions}

In this work, we show that for any constant $k \geq 2$, one can decide whether a graph $G$ is a $k$-leaf power in time $O(n^{f(k)})$.
Here, the function $f$ depends only on $k$, and thus $k$-leaf powers can be recognized in polynomial time for any constant $k$.  We must reckon that $f(k)$ grows faster than a power tower function with base $k$ and height $3k$, i.e. $f(k) \in \Omega(k \uparrow \uparrow (3k))$, using Knuth's up arrow notation.  We did not attempt to optimize $f(k)$, and it is possible that the techniques presented here can be refined in the future to attain a more reasonable $f(k)$ exponent, or even to obtain an FPT algorithm in parameter $k$.

To the best of our knowledge, several tools developed for this result have not been applied before. 
The main idea is that if $G$ is a $k$-leaf power, then either $G$ admits a $k$-leaf root of low maximum degree, in which case it can be found ``easily", or all $k$-leaf roots have large maximum degree, in which case it contains redundant substructures.
By this, we mean that $G$ has a large number of vertex-disjoint subgraphs that are easy to solve individually and admit the ``same kind" of $k$-leaf roots.  We can then argue that we can simply remove one of those redundant subsets from $G$, obtain an equivalent instance, and repeat the process.  We do borrow ideas from~\cite{chen2003computing,eppstein2020parameterized}, since we handle the ``easy" instances mentioned above using dynamic programming on a tree decomposition.

Although the above ideas are used for algorithmic purposes, they may shed light on the graph theoretical characterization of $k$-leaf powers.  Indeed, our collection of similar subgraphs satisfy a number of graph properties of interest (see next subsection).  Combined with the knowledge gained from prior work, our side results may thus help understanding the structure of $k$-leaf powers.  It is also plausible that our techniques can be applied to solve open problems on the larger class of pairwise compatibility graphs~\cite{calamoneri2016pairwise} and their variants, where edges represent a distance in an interval $[d_1, d_2]$ and non-edges represent a distance not in that interval.

\subsection{Overview of our algorithm}

Our approach requires a bit of a setup in terms of definitions, so here we first provide the main intuitions.

Assume for the moment that $G$ admits a $k$-leaf root $T$ of maximum degree $d := d(k)$, some quantity that depends only on $k$.  Then for any leaf $v$ in $T$, there are at most $d^k$ other leaves of $T$ at distance at most $k$ from $v$.  This implies that in $G$, $v$ has at most $d^k$ neighbors, and so the maximum degree of $G$ is at most $d^k$.  This bounds the maximum clique number and, since $G$ is chordal, also bounds the treewidth by $d^k$.  Eppstein and Havvaei~\cite{eppstein2020parameterized} have shown that in this setting, one can decide whether $G$ is a $k$-leaf power in time $O((k d^k)^{c d^k} n)$ for some constant $c$.  

The difficult cases therefore arise when every $k$-leaf root of $G$ has maximum degree above $d$.  
At a very high level, our approach for this case can be described in four essential steps.  

\begin{enumerate}
    \item 
    Find a large collection $\{C_1 \cup Y_1, \ldots, C_d \cup Y_d\}$ of disjoint subsets of $V(G)$ that have a ``similar" neighborhood structure and that are easy to solve.  
    Each $C_i$ is small and cuts $Y_i$ from the rest of the graph, and each subgraph induced by $C_i \cup Y_i$ has maximum degree at most $d^k$.
    
    \item 
    Consider the set of \emph{all} $k$-leaf roots of each subgraph of $G$ induced by the $C_i \cup Y_i$ subsets, and ensure that many of these sets of $k$-leaf roots are  ``similar".
    If $d$ is large enough, this will be the case.
    
    \item 
    Check whether $G - (C_1 \cup Y_1)$ admits a $k$-leaf root.  If not, we are done, but if so, let $T$ be such a $k$-leaf root.
    
    \item 
    Look at how the $C_i \cup Y_i$ subsets are organized within $T$, for $i > 1$.  Since they have a similar structure as $C_1 \cup Y_1$ and admit the same type of $k$-leaf roots, we can find a $k$-leaf root $T_1$ of $G[C_1 \cup Y_1]$ that we can  embed into $T$, while mimicking the organization of the $C_i \cup Y_i$'s.
    
\end{enumerate}

Of course, we need to be precise about what is meant by a ``similar" neighborhood structure, and a ``similar" set of $k$-leaf roots.  We need two ingredients: the notion of a \emph{similar structure of $G$}, and the notion of the \emph{signature} of a tree.

\begin{figure}[t]
    \centering
    \includegraphics[width=\textwidth]{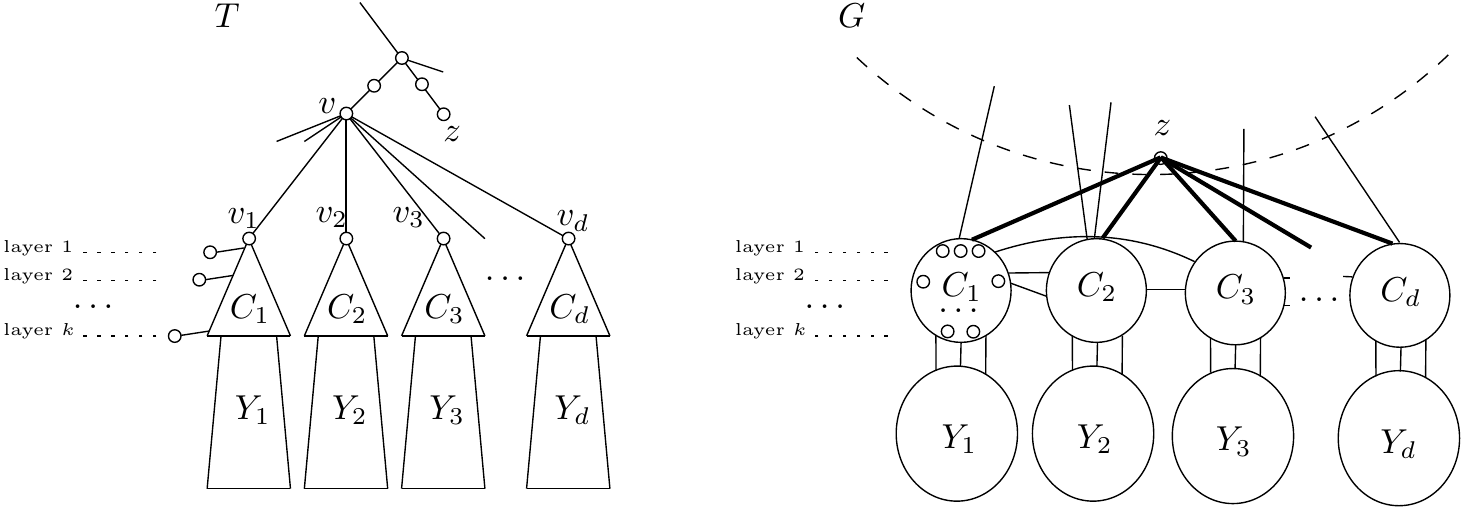}
    \caption{On the left: a $k$-leaf root $T$ with a node $v$ with more than $d$ children.  We have chosen $d$ of them, namely $v_1, \ldots, v_d$.  The leaf $z$ is chosen to be at minimum distance from $v$ in $T$.  The $C_i$ labels represent the leaves descending from $v_i$ at distance at most $k$ from $z$.  Each $C_i$ is organized into at most $k$ layers, where layer $j$ consists of the vertices at distance exactly $j$ from $v$.  The $Y_i$ labels represent the deeper leaves below $v_i$.  
    On the right: the structure of $G$ implied by $z$ and the $C_i \cup Y_i$ leaves of $T$.  In $G$, vertices from the same layer have the same neighbors outside their respective $C_i \cup Y_i$.
    }
    \label{fig:intuition}
\end{figure}

\paragraph{Similar structures.}
Let $T$ be a $k$-leaf root of $G$ of maximum degree above some large enough $d$, and suppose that $T$ is rooted.  
We look at the leaves below a deepest high degree node of $T$, and want to understand how their structure is reflected in $G$.  This is represented in Figure~\ref{fig:intuition}.
More specifically, $T$ has some deepest node $v$ with \emph{at least} $d + 1$ children and whose descendants all have \emph{at most} $d$ children.  Let $v_1, \ldots, v_d$ be $d$ arbitrary children of $v$.  Let $T_1, \ldots, T_d$ denote the subtrees rooted at the $v_i$'s.  Take a leaf $z$ of $T$ that is as close as possible to $v$, but that does not belong to a $T_i$ subtree (such a $z$ exists because $v$ has more than $d$ children).
Then the leaves at distance at most $k$ from $z$ in a $T_i$ subtree form a subset $C_i \subseteq V(G)$ of neighbors of $z$.  Since $T_i$ has maximum degree $d$, one can argue that $|C_i| \leq d^k$.  Moreover, if we assume that $G$ is connected, it can be argued that $C_i \neq \emptyset$.
In fact in $G$, each $C_i$ separates the other leaves contained in $T_i$ from the rest of $G$.  Let us call these other leaves $Y_i$.  Note that $Y_i$ might be empty, so $C_i$ might not exactly be a separator, and even if it is, it might not be minimal.
Nevertheless, the degree bound on $v_i$ and its descendants implies that in $G$, the induced subgraph $G[C_i \cup Y_i]$ has maximum degree $d^k$.
In other words, a high-degree $k$-leaf root of $G$ implies the existence of a large number of disjoint $C_i \cup Y_i$ subsets of vertices that are each ``easy" to solve, and such that there is a   vertex $z$ that is a neighbor of each member of each $C_i$.

There is yet another property of the $C_i$'s that is useful.  If we look at the leaves in $C_i$ at distance $j$ from $v$ for some $j \in [k]$, and the leaves in some other $C_j$ also at distance exactly $j$ from $v$, all these leaves share the same distance to the leaves ``outside" of $T_i$ and $T_j$.  In terms of $G$, the $C_i$ vertices can be ``layered" so that vertices in the same layer have the same outside neighborhood in $G$, where the layer of a vertex is an integer between $1$ and $k$.  
This is what is meant by a similar neighborhood structure.

\paragraph{Exploiting similar structures.}
Obviously, we do not have access to such a $k$-leaf root to find the structure as we just did, even though we know it exists.  However, it can be found in $G$ by brute force, and
our ultimate goal is to show that $G$ is a $k$-leaf power if and only if $G - (C_1 \cup Y_1)$ is a $k$-leaf power (where the choice of $C_1 \cup Y_1$ is arbitrary).

So, imagine that we have found $z$ and the $C_i \cup Y_i$'s as in Figure~\ref{fig:intuition} on the right, along with a layering of the $C_i$'s, but that we do not have the tree $T$ on the left.  In fact, the actual $k$-leaf root might not look like $T$ at all, and it might intertwine its $C_i \cup Y_i$ leaf sets.
Nevertheless, if $d$ is large enough, we show that many of the $G[C_i \cup Y_i \cup \{z\}]$ subgraphs admit a ``similar" set of $k$-leaf roots.  
%
To make this notion clearer, suppose that we take \emph{every} $k$-leaf root of a $G[C_i \cup Y_i \cup \{z\}]$ subgraph, look at their restriction to $C_i \cup \{z\}$, and replace each leaf of $C_i$ by its layer (see Figure~\ref{fig:sig}.c).
By ``similar" sets of $k$-leaf roots, we mean that these sets of restricted $k$-leaf roots are exactly the same for many of the $G[C_i \cup Y_i \cup \{z\}]$ subgraphs.  
These can be computed using dynamic programming on a tree decomposition of $G[C_i \cup Y_i \cup \{z\}]$.

This is an oversimplification, but let us go with it for a moment.
Assuming the above is feasible, we could first find a $k$-leaf root $T$ of $G - (C_1 \cup Y_1)$ (if none exists, then $G$ is not a $k$-leaf power).  If $d$ is large enough, several of the $C_i \cup \{z\}$ subsets will be organized in the same manner in $T$, i.e. restricting $T$ to these $C_i \cup \{z\}$ subsets and replacing $C_i$ leaves by their layer will yield the same tree.  We can then find a $k$-leaf root $T_1$ of $G[C_1 \cup Y_1 \cup \{z\}]$ with the same structure as these, and embed $T_1$ with the same organization as the others in $T$.  Because $C_1$ is layered in the same manner as the other $C_i$'s, mimicking their structure  in $T$ ensures that the distance relationships with the other vertices are satisfied.
The $z$ vertex is important, since it serve as a common point between all the $C_i$'s and indicates where to start embedding.
At this point, it becomes difficult to describe how this embedding is performed more concretely, 
and the interested reader is redirected to the proof of Theorem~\ref{thm:iffc1} for more details.

\paragraph{Tree signatures.}  
As we mentioned, expecting the $G[C_i \cup Y_i \cup \{z\}]$ subgraphs to admit exactly the same sets of $k$-leaf roots, even if we restrict them to $C_i \cup \{z\}$ and replace the leaves by their layer, does not quite work.  
First, each restricted subtree must also remember how far the $Y_i$ vertices are from the nodes of the restricted subtree, to ensure that we do not make the $Y_1$ nodes too close to other nodes during the embedding.
This can be done by labeling the internal nodes of our restricted trees with the distance to the closest $Y_i$ leaf, as in Figure~\ref{fig:sig}.c.  Note that a similar trick was done by Chen et al. in~\cite{chen2003computing}.

\begin{figure}
    \centering
    \includegraphics[width=\textwidth]{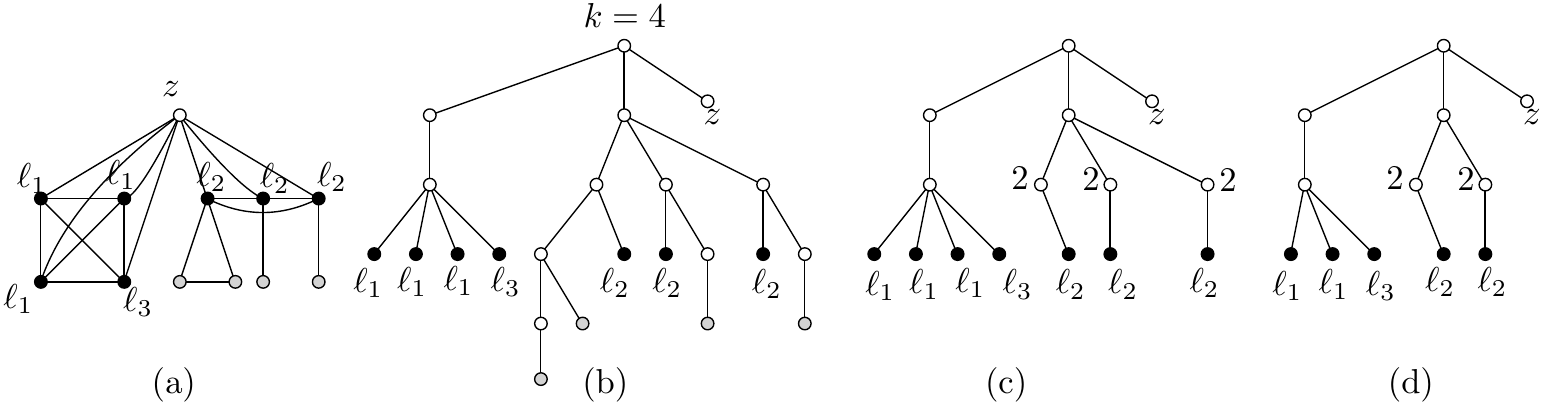}
    \caption{(a) A subgraph of $G$ induced by $C_i \cup Y_i \cup \{z\}$.  Vertices in black are those of $C_i$, and in gray those of $Y_i$.  The $\l_1, \l_2, \l_3$ represent the layer assigned to each vertex of $C_i$ (several vertices can have the same layer).  (b) A $4$-leaf root $T$ for this subgraph.  The leaves of $C_i$ are labeled by their layer.  (c) The restriction of $T$ to $C_i \cup \{z\}$, but in which each internal node remembers the distance to the closest pruned leaf that it led to (if any).  (d) A tree representing the signature of the restricted $4$-leaf root.  Each time a node has three identical subtrees, one of them is considered redundant and is pruned.}
    \label{fig:sig}
\end{figure}

Second, we cannot guarantee that enough of the $G[C_i \cup Y_i \cup \{z\}]$ will admit \emph{exactly} the same set of restricted $k$-leaf roots, since the number of possibilities is too large.
The solution is to define a compact representation of $k$-leaf roots restricted to some $C_i \cup \{z\}$, which we call its \emph{signature}.  
This representation prunes information from the tree that is not necessary for our embedding.  Conceptually, to obtain the signature of a tree, we look at each node and if we find a node that has three or more child subtrees that are identical, we remove one of these subtrees.  This is fine for our embedding, since the identical subtree can be inserted along the others that are identical.  We repeat until this is not possible.  
This is illustrated in Figure~\ref{fig:sig}.  From (a) a $G[C_i \cup Y_i \cup \{z\}]$ subgraph, we look at (b) each $k$-leaf root with leaves of the $C_i$'s replaced by their layer, then (c) restrict to $C_i \cup \{z\}$ and remember the distances to the removed leaves, and finally (d) prune redundant subtrees.  This gives the compact representation of one $k$-leaf root, and we must obtain them all.

This compact representation allows much less possibilities, and then we can guarantee with a pigeonhole argument that with large enough $d$, many $G[C_i \cup Y_i \cup \{z\}]$ will admit the exact same set of signatures.  This works because the number of possible signatures does not depend on $d$, only on $k$.  Therefore, we can make $d$ as large as desired for our argument.
Note that concretely, signatures are represented as vectors of integers that encode the same information, since it allows for simpler proofs (see Section 3).

We now proceed with the details.

\section{Preliminary notions}

For an integer $n$, we use the notation $[n] = \{1, 2, \ldots, n\}$.
All graphs in this paper are finite, simple and undirected.  For a graph $G$ and $v \in V(G)$, we denote by $N_G(v)$ the set of neighbors of $v$ in $G$, and we write $N_G[v] = N_G(v) \cup \{v\}$.  
For $X \subseteq V(G)$, we denote by $N_G(X) = \bigcup_{x \in X} (N_G(x) \setminus X)$ the neighbors of members of $X$ that are outside of $X$.  Also, we write $N_G[X] = N_G(X) \cup X$.
We may drop the subscript $G$ if it is clear from the context.  For $X \subseteq V(G)$, we denote by $G[X]$ the subgraph of $G$ induced by $X$.
We define a \emph{connected component} as a maximal \emph{set} of vertices $X$ such that $G[X]$ is connected.  

Unless stated otherwise, all trees in this paper are \emph{rooted}.  Hence we will usually say \emph{tree} instead of \emph{rooted tree}.
We denote the root of a tree $T$ by $r(T)$.  
For a node $v \in V(T)$, we write $ch_T(v)$ for the set of children of $v$ in $T$.  The \emph{arity} of $T$ is $\max_{v \in V(G)}|ch_T(v)|$.
We say that $v$ is a \emph{leaf} if it is has no children.
We write $L(T)$ to denote the set of leaves of $T$.
It is important to note that leaves are sometimes defined as nodes with a single neighbor in $T$.  This slightly differs from our definition, since if $r(T)$ has a single child in $T$, it is not treated as a leaf here.  
The only case in which the root is also a leaf is when $T$ has a single vertex (which has no children).
A node $u \in V(T)$ is a \emph{descendant} of another node $v \in V(T)$ if $v$ is on the path from $r(T)$ to $u$.  In this case, $v$ is an \emph{ancestor} of $u$.  Note that $v$ is a descendant and ancestor of itself.
Given a tree $T$ and some $v \in V(T)$, we let $T(v)$ denote the subtree rooted at $v$, i.e. the subgraph of $T$ induced by $v$ and all its descendants.
The distance between two nodes $u$ and $v$ of $T$ is denoted $dist_T(u, v)$.
We define $height(T) = 1 + \max_{l \in L(T)} dist_T(r(T), l)$.

Two trees $T_1$ and $T_2$ are \emph{equal} if $r(T_1) = r(T_2)$ and $(V(T_1), E(T_1)) = (V(T_2), E(T_2))$, in which case we write $T_1 = T_2$.  Two trees $T_1$ and $T_2$ are called \emph{leaf-isomorphic}, denoted $T_1 \simeq_L T_2$, if $L(T_1) = L(T_2)$, and there exists a bijection $\mu : V(T_1) \rightarrow V(T_2)$ such that $\mu(u) = u$ for every $u \in L(T_1)$, $\mu(r(T_1)) = r(T_2)$, and such that $uv \in E(T_1)$ if and only if $\mu(u)\mu(v) \in E(T_2)$.  We call $\mu$ a leaf-isomorphism.
Note that this is stronger than the usual notion of isomorphism, since we require $T_1$ and $T_2$ to be built with the same set of leaves.  Also observe that since leaves must be matched, the $\mu$ function is unique\footnote{To see this, observe that $\mu$ is forced for the leaves of $T_1$.  Then, $\mu$ is forced for the parent of those leaves, and then $\mu$ is forced for the grand-parents, and so on.}.

Let $X \subseteq V(T)$.
The \emph{restriction} of $T$ to $X$, denoted $T|X$, is the subgraph of $T$ induced by $X$ and every vertex of $T$ on the shortest path between two elements of $X$.  We shall repeatedly use the fact that for $u, v \in X$, $dist_{T|X}(u, v) = dist_T(u, v)$, and that
$V(T|X) \subseteq V(T)$.
Note that for $v \in V(T|X)$, $ch_{T}(v) \setminus ch_{T|X}(v)$ denotes the set of children of $v$ that were ``removed" by the restriction.

\subsection*{Leaf powers and their properties}

The main definition of interest in this paper is the following.

\begin{definition}
Let $G$ be a graph and let $k$ be a positive integer.  A \emph{$k$-leaf root} of $G$ is a tree $T$ such that $L(T) = V(G)$, and such that for all distinct $u, v \in V(G)$,  $uv \in E(G)$ 
if and only if $dist_T(u, v) \leq k$.  

The graph $G$ is a \emph{$k$-leaf power} if there exists a $k$-leaf root of $G$.
\end{definition}

Note that according to our definitions, the tree $T$ is implicitly rooted.
This is not required in the usual definition of leaf powers, although this is inconsequential for $k \geq 2$.  This is because if an unrooted $k$-leaf root $T$ has an internal node or has a single node, we may use it as the root.  If not, $T$ has two vertices $u$ and $v$, which are both of degree $1$.  This implies that $uv \in E(G)$, but since $k \geq 2$, we may add an internal node between $u$ and $v$ in $T$.
Also observe that the $k$-leaf power property is hereditary.  That is,
if $G$ is a $k$-leaf power and $X \subseteq V(G)$, then $G[X]$ is a $k$-leaf power.  This is because if $T$ is a $k$-leaf root of $G$, then $T|X$ is a $k$-leaf root of $G[X]$, since restrictions preserve distances.

A graph is \emph{chordal} if it has no induced cycle of length $4$ or more.  The treewidth of $G$ is denoted $tw(G)$ (we defer the definition of treewidth to Section~\ref{sec:tw}).
The following is a well-known property of $k$-leaf powers, for any $k$ (see~\cite{nishimura2002graph,brandstadt2010rooted}).

\begin{lemma}\label{lem:chordal}
Let $G$ be a $k$-leaf power.  Then $G$ is chordal.
\end{lemma}

We will be interested in $k$-leaf roots that have high arity.  
In~\cite[section 7]{eppstein2020parameterized}, Eppstein and Havvaei have shown that if $G$ is a graph of bounded treewidth, then deciding if $G$ is a $k$-leaf power is fixed-parameter tractable in $k + tw(G)$.
As we show, this implies that low-arity $k$-leaf roots, if any, can be found using this result.

\begin{lemma}\label{lem:prelim:bounded-arity}
Let $G$ be a graph with $n$ vertices.  If there exists a $k$-leaf root of $G$ of arity at most $d$, then $G$ has maximum degree $d^k$.  
Moreover, one can decide in time $O(n (d^k k)^{c d^k})$ whether such a $k$-leaf root exists, where $c$ is a constant.
\end{lemma}

\begin{proof}
Assume that $T$ is a $k$-leaf root of $G$ of arity at most $d$.  
Let $v \in V(G)$.  Then in $T$, the number of leaves at distance at most $k$ from $v$ is bounded by $d^k$ (since if we imagine rerooting $T$ at $v$, each node will still have at most $d$ children --- the astute reader will see that we could optimize to $d^{k-1}$, but let us not bother).  This implies that $v$ has at most $d^k$ neighbors in $G$.  Since this holds for every vertex, $G$ has maximum degree $d^k$.

As for the second part of the lemma, we know by~\cite{eppstein2020parameterized} that deciding if $G$ is a $k$-leaf power can be done in time $O(n \cdot (tw(G) k)^{c \cdot tw(G)})$ for some constant $c$. To use this, we can first check whether $G$ is chordal in linear time, and if not reject it.  Since $G$ is chordal, it is well-known that $tw(G) = \omega(G) - 1$, where $\omega(G)$ is the size of a maximum clique in $G$.  In any graph, $\omega(G)$ is at most the maximum degree plus one.  In our case, this implies that $tw(G) = \omega(G) - 1 \leq d^k$.
Using the algorithm of~\cite{eppstein2020parameterized}, we can check whether $G$ is a $k$-leaf power in time $O(n (d^k k)^{c d^k})$.
\end{proof}

\section{Finding redundant structures in $k$-leaf powers}

Let us fix a positive integer $k$ and an arbitrary graph $G$ for the rest of the paper.  We assume that $G$ is connected, as otherwise, each connected component can be treated separately (if a $k$-leaf root is found for each component, we can join their roots under a new root at distance more than $k$ from them).
As we mentioned, an important difficulty is when $G$ does admit $k$-leaf roots, but they all have large arity.  
We now define our similar structures precisely.
We then introduce the notion of a \emph{signature} for their $k$-leaf roots.  After that, we argue that many subsets admit the same $k$-leaf root signatures, and that we can prune one.
In Section~\ref{sec:tw}, we show how the set of $k$-leaf root signatures can be found.

\subsection{Similar structures}

A \emph{similar structure} of a graph $G$ is a tuple $\S = (\C, \Y, z, \L)$ where:
\begin{itemize}
    \item 
    $\C = \{C_1, \ldots, C_d\}$ is a collection of $d \geq 2$ pairwise disjoint,  non-empty subsets of vertices of $G$;
    
    \item 
    $\Y = \{Y_1, \ldots, Y_d\}$ is a collection of pairwise disjoint subsets of vertices of $G$, some of which are possibly empty.  Also, $C_i \cap Y_j = \emptyset$ for any $i, j \in [d]$;
    
    \item 
    $z \in V(G)$ and does not belong to any subset of $\C$ or $\Y$;
    
    \item 
    $\L = \{\l_1, \ldots, \l_d\}$ is a set of functions where, for each $i \in [d]$, we have $\l_i: C_i \cup \{z\} \rightarrow \{0, 1, \ldots, k\}$.
    The functions in $\L$ are called \emph{layering functions}.
    
\end{itemize}

Additionally, $\S$ must satisfy several conditions.  
Let us denote $C^* = \bigcup_{i\in [d]} C_i$.
Let $X = \{X_1, \ldots, X_t\}$ be the connected components of $G - C^*$.
For each $i \in [d]$, denote $X^{(i)} = \{X_j \in X : N_G(X_j) \subseteq C_i\}$, i.e. the components that have neighbors only in $C_i$.

Then all the following conditions must hold:

\begin{enumerate}
    \item  \label{cut:yi}
    for each $i \in [d]$, $Y_i = \bigcup_{X_j \in X^{(i)}} X_j$ ($Y_i = \emptyset$ is possible);

    \item \label{cut:znbrhood}
    there is exactly one connected component $X_z \in X$ such that for all $i \in [d]$,
    ${N_G(X_z) \cap C_i \neq \emptyset}$.  Moreover, $z \in X_z$ and $C^* \subseteq N_G(z)$;
    
    \item \label{cut:ccs}
    for all $X_j \in X \setminus \{X_z\}$, $X_j \subseteq Y_i$ for some $i \in [d]$.  In particular, $X_z$ is the only connected component of $G - C^*$ with neighbors in two or more $C_i$'s;

    \item  \label{cut:layers}
    the layering functions $\L$ satisfy the following:
    \begin{enumerate}
        
        
        \item \label{cut:zlayer}
        for each $i \in [d]$, $\l_i(z) = 0$.  Moreover, $\l_i(x) > 0$ for any $x \in C_i$;
        
        \item \label{cut:layerssame}
        for any $i, j \in [d]$ and any $x \in C_i, y \in C_j$, $\l_i(x) = \l_j(y)$ implies $N_G(x) \setminus (C_i \cup Y_i \cup C_j \cup Y_j) = N_G(y) \setminus (C_i \cup Y_i \cup C_j \cup Y_j)$.  
        Note that this includes the case $i = j$;
        
        \item \label{cut:layers-leq-k}
        for any $i, j \in [d]$ and any $x \in C_i, y \in C_j$, $\l_i(x) + \l_j(y) \leq k$ implies $xy \in E(G)$.  Note that this includes the case $i = j$.
        
        \item \label{cut:layers-gt-k}
        for any \emph{two distinct} $i, j \in [d]$ and any $x \in C_i, y \in C_j$, $\l_i(x) + \l_j(y) > k$ implies $xy \notin E(G)$.  Note that this does \emph{not} include the case $i = j$
        
    \end{enumerate}

\end{enumerate}

We will refer to the value of $d$ as the \emph{size} of $\S$.

Although somewhat burdensome, the properties of a similar structure occur naturally when we look at the subtrees under a given node of a $k$-leaf root.  The similar structure is the one that is described in Figure~\ref{fig:intuition}, but in a more precise manner.  
The properties of similar structures essentially say that after removing each $C_i \in \C$ from $G$, there is one connected component $X_z$ that touches every $C_i$, with a $z \in X_z$ that is a neighbor of each $C_i$.  All the other connected components are separated from the rest of the graph by exactly one $C_i$, and these form the $Y_i$'s.  As for the layering functions, they state that $z$ is a special vertex with layer $0$.  
For the $C_i$ vertices, the layers represent how the neighborhoods of the $C_i$ members are organized.  One can imagine that $G$ has a $k$-leaf root and that the layer of $x \in C_i$ is the distance from $x$ to the lowest common ancestor of $C_i$ (of course, this is conceptual since we don't have this $k$-leaf root).
Any two vertices at the same layer must have the same neighborhood outside of their $C_i$'s.  Vertices from ``close" layers, i.e. with sum at most $k$, must be neighbors.  If the layers are ``far", i.e. with sum more than $k$, then the vertices should not be neighbors (unless they are in the same $C_i$).

We first show that similar structures can always be found on graphs with $k$-leaf roots of high arity.
This is essentially a formalization of Figure~\ref{fig:intuition}.

\begin{lemma}\label{lem:exists-similar}
Let $d \geq 2$ and let $G$ be a connected graph that admits a $k$-leaf root of arity at least $d + 1$.
Then there exists a similar structure $\S = (\C, \Y, z, \L)$ of $G$ such that 
$|\C| = d$.  Moreover, for each $C_i \in \C$, $|C_i| \leq d^k$
and $G[C_i \cup Y_i \cup \{z\}]$ has maximum degree at most $d^k$.
\end{lemma}

\begin{proof}
Let $T$ be a $k$-leaf root of $G$ of arity at least $d + 1$.  We show how to construct $\C = \{C_1, \ldots, C_d\}$, $\Y = \{Y_1, \ldots, Y_d\}, z$, and $\L= \{\l_1, \ldots, \l_d\}$ using the relationship between $T$ and $G$.
Let $v$ be a deepest node of $T$ with $d + 1$ or more children (such a node must exist).  By the choice of $v$, all the descendants of $v$ have at most $d$ children.  

Let $z \in L(T)$ be a leaf of $T$ at minimum distance from $v$, i.e. $z$ minimizes $dist_T(z, v)$ among all leaves of $T$.
Note that $z$ may or may not be a descendant of $v$.
Let $v_1, \ldots, v_d$ be $d$ children of $v$, none of which is an ancestor of $z$ (such a choice of $d$ children exists since $v$ has at least $d + 1$ children).
For each $i \in [d]$, define $C_i = L(T(v_i)) \cap N_G(z)$.
Note that since $z$ is not a descendant of $v_i$, the vertices of $L(T(v_i)) \cap N_G(z)$ must be at distance at most $k$ from $v$ (actually, $k - 1$, but we shall not bother).  Since the $T(v_i)$ subtree has arity at most $d$, there are at most $d^k$ leaves in $T(v_i)$ at distance at most $k$ from $v$.  It follows that $|C_i| \leq d^k$, for each $i \in [d]$.
Note that by construction, the $C_i$ subsets are pairwise disjoint.  To see that the $C_i$'s are non-empty, assume that $C_i$ is empty for some $i \in [d]$.
Since $z$ is the closest leaf to $v$, and since every path from a member of $L(T(v_i))$ to a member of $V(G) \setminus L(T(v_i))$ passes through $v$, no element of $L(T(v_i))$ is at distance $k$ or less in $T$ from and element of $V(G) \setminus L(T(v_i))$.  Since $L(T(v_i)) \neq \emptyset$, this implies that $G$ is disconnected, a contradiction.  We may thus assume that $C_i \neq \emptyset$.

Now, for convenience, define $C^* = \bigcup_{i \in [d]} C_i$.
Also define $G' = G - C^*$.

Let $X_z$ be the connected component of $G'$ that contains $z$.  By construction, $z$ is a neighbor of every vertex in each $C_i$.  We must show that only $X_z$ intersects with every $C_i$.
Let $Z = L(T) \setminus \left( \bigcup_{i \in [d]} L(T(v_i))) \right)$, 
and notice that $z \in Z$.
We argue that $G'[Z]$ is connected.
Assume otherwise.  Then $G'[Z]$ has at least two connected components, one of them being $X_z$, and the other some other component $X_q$.  Since $G$ is connected and removing $C^*$ separates $X_z$ from $X_q$, there is some $q \in X_q$ such that $N_G(q) \cap C^* \neq \emptyset$.
Let $c \in N_G(q) \cap C^*$.
Consider the location of $q$ in $T$.  Because $q \in Z$, $q$ is not a descendant of any $v_i$, and thus the path from $q$ to $c$ in $T$ passes through $v$.  But the choice of $z$ implies 
\begin{align*}
k \geq dist_T(q, c) = dist_T(q, v) + dist_T(v, c) \geq dist_T(q, v) + dist_T(v, z) \geq dist_T(q, z)
\end{align*}
Since $T$ is a $k$-leaf root of $G$, we have $qz \in E(G)$, contradicting that they belong to different connected components of $G'$.
Thus $G'[Z]$ is connected. 

Now for $i \in [d]$, consider a leaf $x \in L(T(v_i)) \setminus C_i$.
We argue that all neighbors of $x$ are in $L(T(v_i))$.
Suppose that $x$ has a neighbor $q \in L(T) \setminus L(T(v_i))$.  
Then $dist_T(x, q) \leq k$, and the path from $x$ to $q$ goes through $v$.  But again, by the choice of $z$, 
\[
k \geq dist_T(x, q) = dist_T(x, v) + dist_T(v, q) \geq dist_T(x, v) + dist_T(v, z) = dist_T(x, z)
\]
  But then, $x$ should be in $C_i$ (by its definition), which is a contradiction.  It follows that $x$ has no neighbor outside of $L(T(v_i))$.
In particular, $x \notin Z$.

It follows that $Z$ is a connected component of $G'$, that it is equal to $X_z$, and that it satisfies Property~\ref{cut:znbrhood} of similar structures.
Moreover, as we argued any other connected component $X_j$ of $G'$ contains vertices that appear as leaves of some $T(v_i)$ subtree, and $N_G(X_j)$ must be a subset of $C_i$ since $X_j$ vertices have no neighbor outside $T(v_i)$. 
We may then define $Y_i = L(T(v_i)) \setminus C_i$ for each $i \in [d]$ and it follows that $Y_i$ regroups all connected components of $G'$ that have neighbors only in $C_i$, which satisfies Property~\ref{cut:yi}.
Property~\ref{cut:ccs} is easily seen to be satisfied, since the elements in $\{X_z, Y_1, \ldots, Y_d\}$ are disjoint and cover all of $L(T) \setminus C^*$.


It remains to describe our layering functions $\L$.  
For $i \in [d]$, put $\l_i(z) = 0$  and for $x \in C_i$, define
\[
\l_i(x) = dist_T(x, v)
\]
Note that $\l_i(x) > 0$, and that $\l_i(x) = 1$ is possible if $C_i$ only contains a leaf of $T$. 
Also note that $\l_i(x) \leq k$ since each $x$ is a neighbor of $z$.
Property~\ref{cut:zlayer} is thus satisfied.  We show that each remaining property is satisfied.

\emph{Property~\ref{cut:layerssame}}.  
Let $i, j \in [d]$ and let $x \in C_i, y \in C_j$ such that $\l_i(x) = \l_j(y)$.  
Note that $q \in N_G(x) \setminus (C_i \cup Y_i \cup C_j \cup Y_j)$ only if $q$ is a leaf of $T$ at distance at most $k$ from $x$, and $q$ does not descend from $v_i$ or $v_j$.  Therefore, the path from $x$ to $q$ goes through $v$.
Since $dist_T(x, v) = dist_T(y, v)$ (because they have the same layer), we have $dist_T(x, q) = dist_T(y, q)$, and thus $q \in N_G(y) \setminus (C_i \cup Y_i \cup C_j \cup Y_j)$ as well.  Thus $N_G(x) \setminus (C_i \cup Y_i \cup C_j \cup Y_j) \subseteq N_G(y) \setminus (C_i \cup Y_i \cup C_j \cup Y_j)$, and the other containment direction can be argued symmetrically.  Thus $\S$ satisfies Property~\ref{cut:layerssame}.

\emph{Property~\ref{cut:layers-leq-k}}.  Let $i, j \in [d]$ and let $x \in C_i, y \in C_j$ with $\l_i(x) + \l_j(y) \leq k$. We see that $dist_T(x, y) \leq dist_T(x, v) + dist_T(y, v) = \l_i(x) + \l_j(y) \leq k$, implying that $xy \in E(G)$.  Thus Property~\ref{cut:layers-leq-k} is satisfied.

\emph{Property~\ref{cut:layers-gt-k}}.  Let $i, j \in [d]$, $i \neq j$, and let $x \in C_i, y \in C_j$ with $\l_i(x) + \l_j(y) > k$.  Since $x$ and $y$ descend from distinct $v_i$ and $v_j$, $dist_T(x, y) = dist_T(x, v) + dist_T(v, y) = \l_i(x) + \l_j(y) > k$, and thus $xy \notin E(G)$.
Thus Property~\ref{cut:layers-gt-k} is satisfied.

We have therefore shown that $\S$ satisfies all requirements to be a similar structure, and that $|C_i| \leq d^k$ for each $i \in [d]$.

To finish the proof, recall that the lemma requires that each $G[\hC{i}]$ has maximum degree at most $d^k$.  First note that $z$ has at most $d^k$ neighbors in this subgraph.
Now let $x \in C_i \cup Y_i$.  Any neighbor of $x$ in $G[\hC{i}]$ is at distance at most $k - 1$ from the parent $p$ of $x$ in $T$.  Since $T_i$ has arity at most $d$, the number of other leaves of $C_i \cup Y_i$ at distance at most $k - 1$ from $p$ is at most $d^{k-1}$, and thus $x$ has at most $d^{k-1}$ neighbors in $C_i \cup Y_i$ (this also counts the leaves that do not descend from $p$).  Note that $x$ could also have $z$ as a neighbor, so its degree in $G[\hC{i}]$ is at most $d^{k-1} + 1 \leq d^k$, as desired.  
\end{proof}

\subsection{Valued trees and signatures}

We now know that similar structures exist in difficult $k$-leaf powers.  This is not enough though.  We would like to say that the $G[C_i \cup Y_i \cup \{z\}]$ subgraphs each admit ``similar'' sets of $k$-leaf roots, so that pruning one $C_i \cup Y_i$ does not matter.  In fact, we are only interested in how the $C_i \cup \{z\}$ subsets behave in these $k$-leaf roots, so we develop the notion of a \emph{valued restriction} to remove the $Y_i$'s, while retaining the essential distance information to these $Y_i$'s.  This is still not enough though, since the sets of such restricted $k$-leaf roots may still differ.  We introduce the notion of a \emph{signature} for these pruned trees, which is a compact representation that retains the essence of the structure of the trees.

A \emph{valued tree} $\T$ is a pair $(T, \s)$ where $T$ is a tree and $\s : V(T) \setminus L(T) \rightarrow \mathbb{N} \cup \{\infty\}$ assigns each internal node of $T$ an integer, or possibly the special value $\infty$.  
We say that $\T$ is \emph{$s$-bounded} if $\s(v) \leq s$ or $\s(v) = \infty$ for each $v \in V(T) \setminus L(T)$.  We define $height(\T) = height(T)$.
For $v \in V(T)$, $\T(v) = (T', \s')$ denotes the valued tree in which $T' = T(v)$ and $\s'(w) = \s(w)$ for $w \in V(T(v)) \setminus L(T)$.  
We say that two valued trees $(T_1, \s_1)$ and $(T_2, \s_2)$ are \emph{value-isomorphic} if $T_1 \simeq_L T_2$, with leaf-isomorphism $\mu$, and $\s_1(w) = \s_2(\mu(w))$ for all $w \in V(T_1) \setminus L(T_1)$.
The notion of valued restrictions will be fundamental for the rest of this paper.

\begin{definition}
Let $T$ be a tree and let $X \subseteq L(T)$.  
We say that $(T', \s)$ is the \emph{valued restriction of $T$ to $X$} if it satisfies the following:

\begin{itemize}
    \item
    $T' = T|X$;

    \item 
    for each $v \in V(T') \setminus L(T')$, 
    let $L^*(v) = \bigcup_{x \in ch_T(v) \setminus ch_{T'}(v)} L(T(x))$.  Then
    \begin{itemize}
    \item 
    if $L^*(v) = \emptyset$, $\s(v) = \infty$;
    
    \item 
    otherwise, $\s(v) = \min_{l \in L^*(v)} dist_T(v, l)$. 
    \end{itemize}
\end{itemize}
\end{definition}

See Figure~\ref{fig:sig2}.b for an illustration.  Intuitively, the valued restriction of $T$ to $X$ takes $T' = T|X$ as a tree.
By doing so, each remaining $v \in V(T')$ might have lost some children that were in $T$.  We want $v$ to remember how far it is from a ``hidden" leaf, i.e. descending from one of these lost children, and $\s(v)$ stores this information.
The point of this is that $T$ represents a $k$-leaf root in which the leaves not in $X$ should be at distance more than $k$ from any leaf not in $T$ (concretely, these will be the $Y_i$ leaves).  The tree $T'$ is a compressed version of $T$, and we will embed $T'$ into a larger $k$-leaf root instead of $T$.  During this embedding we want to know how far the hidden leaves are to ensure that outside leaves remain at distance more than $k$ from them.

\begin{figure}
    \centering
    \includegraphics[width=\textwidth]{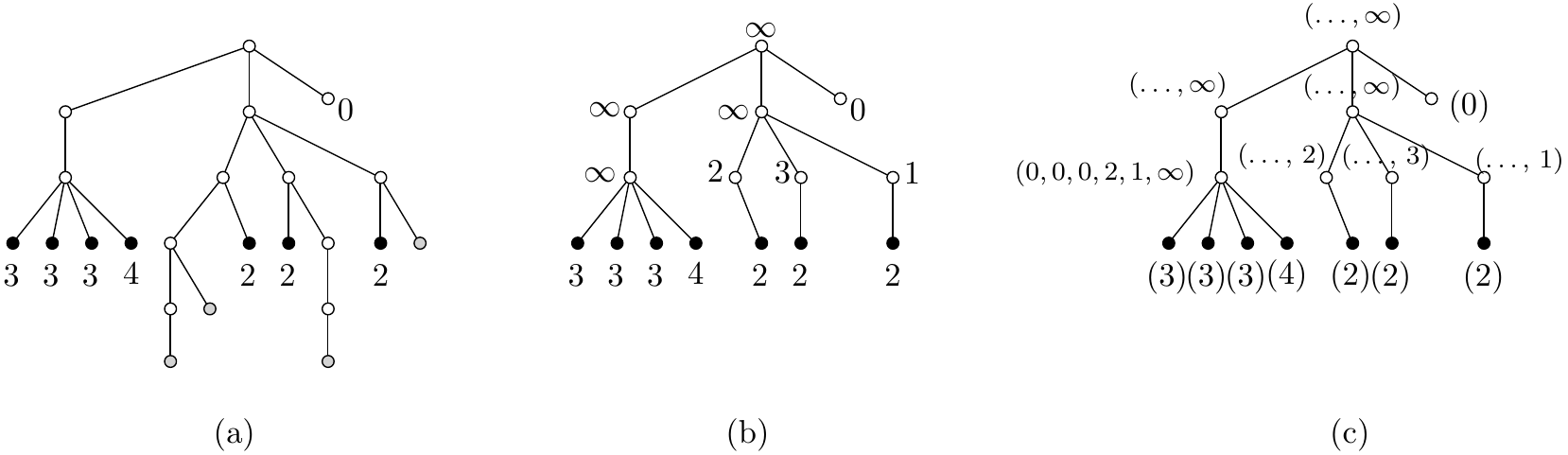}
    \caption{(a) A tree $T$.  The set of leaves $v$ with a label $\l(v)$ in $\{0, 1, \ldots, 4\}$ forms a set $X$, and the leaves in gray are those not in $X$.  Note that these labels do not take part of the definition of a valued restriction, but are needed for signatures. (b) The valued restriction $\T$ of $T$ to $X$, with the leaves of $X$ preserving their labels.  (c) An illustration of the signature of each rooted subtree of the valued restriction.  For each leaf $v$, the signature of $\T(v)$ is $(\l(v))$.  We only give the full signature for one internal node.  We chose to order its entries such that for $i \in 1, \ldots, 5$, the $i$-th coordinate is the number of children of that node that have signature $(i - 1)$, or $2$ if this value is above $2$.  The last entry is the $\s$ value of the node (as for every other node).  The parent of that node would then have one entry for each possible signature for a tree of height $2$, and only one of those entries would be set to $1$.}
    \label{fig:sig2}
\end{figure}

\subsubsection*{Valued tree signatures}

Let $\T = (T, \s)$ be a valued tree that is $s$-bounded for some $s$.  Furthermore, let $\l$ be a layering function that maps each leaf in $L(T)$ to an integer in $\{0, 1, \ldots, k\}$.
We now define the \emph{signature of $\T$ with respect to $\l$}, denoted $sig_{\l}(\T)$, which encodes some properties of $\T$ in a vector of integers.  The signature is defined recursively and depends on the height of $\T$, as illustrated in Figure~\ref{fig:sig2}.c.

If $height(\T) = 1$, then $T$ has a single node $v$, which is a leaf. In this case, define $sig_{\l}(\T) = (\l(v))$.

Now, assume that $height(\T) = h > 1$.  Let $S(s, h - 1) = \{s_1, \ldots, s_{m}\}$ denote the set of all possible signatures for an $s$-bounded valued tree of height $h - 1$ \emph{or less}, with respect to any layering function that assigns leaves to $\{0, 1, \ldots, k\}$.  We will later show that $m$ is finite.  We may assume that the $s_i$ subscripts order the the signatures lexicographically (but this is merely for concreteness, the ordering does not matter).  
Then $sig_{\l}(\T)$ is a vector or dimension $m + 1$ in which, for $i \in [m]$, the $i$-th coordinate of $sig_{\l}(\T)$ is defined as
\[
sig_{\l}(\T)[i] = \min(2, |\{v \in ch_T(r(T)) : sig_{\l}(\T(v)) = s_i \}|)
\]
Moreover, $sig_{\l}(\T)[m + 1] = \s(r(T))$.

It should be obvious that if $\T$ has height $h$, then $sig_{\l}(\T(v)) \in S(s, h - 1)$ for each $v \in ch_{T}(r(T))$, and  hence the signature summarizes every child subtree of $r(T)$.
Two signatures are \emph{equal} in the usual sense of vector equality, i.e. if they have the same length and, at every position, the values of the two vectors are equal.

In words, one may think of each integer $i \in [m]$ as a code for the $i$-th signature $s_i$, and $sig_{\l}(\T)[i]$ is the number of children of $r(T)$ whose subtree has signature $s_i$, except that we only bother remembering whether there are $0, 1$ or more of those (which we encode by $2$).  
We also reserve the last coordinate of $sig_{\l}(\T)$ for $\s(r(T))$. 
For valued trees with one node, we only need to remember the layer of the leaf.

We list the basic properties of signatures that will be needed, and then we will proceed with our $k$-leaf power algorithm.

\begin{lemma}\label{lem:basic-sigs}
Let $\T_1 = (T_1, \s_1)$ and $\T_2 = (T_2, \s_2)$ be valued trees satisfying $sig_{\l_1}(\T_1) = sig_{\l_2}(\T_2)$ for some layering functions $\l_1$ and $\l_2$.
Then the following holds:
\begin{enumerate}
    
    \item  \label{lem:basic-onechild}
    if $r(T_1)$ has a child $u$ such that $sig_{\l_1}(\T_1(u)) \neq sig_{\l_1}(\T_1(v))$ for every $v \in ch_{T_1}(r(T_1)) \setminus \{u\}$, then $r(T_2)$ has exactly one child $u'$ such that $sig_{\l_1}(\T_1(u)) = sig_{\l_2}(\T_2(u'))$;
    
    \item  \label{lem:basic-twochild}
    if $r(T_1)$ has two distinct children $u, v$ such that $sig_{\l_1}(\T_1(u)) = sig_{\l_1}(\T_1(v))$, then 
    $r(T_2)$ has at least two distinct children $u', v'$ such that 
    $sig_{\l_1}(\T_1(u)) = sig_{\l_2}(\T_2(u')) = sig_{\l_2}(\T_2(v'))$;
    
    \item  \label{lem:basic-samex}
    for any $x \in V(T_1)$, there exists $y \in V(T_2)$ such that $dist_{T_1}(x, r(T_1)) = dist_{T_2}(y, r(T_2))$ and $sig_{\l_1}(\T_1(x)) = sig_{\l_2}(\T_2(y))$.  In particular, this includes the case $x \in L(T_1)$. 
    
\end{enumerate}
\end{lemma}

\begin{proof}
Property~\ref{lem:basic-onechild} is due to the fact that the entry corresponding to $sig_{\l_1}(\T_1(u))$ must be equal to $1$ in both the $sig_{\l_1}(\T_1)$ and $sig_{\l_2}(\T_2)$ vectors.  Property~\ref{lem:basic-twochild} is because the entry corresponding to $sig_{\l_1}(\T_1(u)) = sig_{\l_1}(\T_1(v))$ must be equal to $2$ in both the $sig_{\l_1}(\T_1)$ and $sig_{\l_2}(\T_2)$ vectors.

Property~\ref{lem:basic-samex} can be argued inductively on $dist_{T_1}(x, r(T_1))$.  If the distance is $0$, then $x = r(T_1)$ and $y = r(T_2)$ shows that the property holds.
If $dist_{T_1}(x, r(T_1)) > 0$, let $x'$ be the child of $r(T_1)$ on the path between $x$ and $r(T_1)$.  
By the previous properties, $r(T_2)$ has a child $y'$ such that $sig_{\l_1}(\T_1(x')) = sig_{\l_2}(\T_2(y'))$.
By induction, there is $y \in V(T_2(y'))$ such that $dist_{T_1(x')}(x, x') = dist_{T_2(y')}(y, y')$ and $sig_{\l_1}(\T_1(x')) = sig_{\l_2}(\T_2(y'))$.
The distance from $x$ to $r(T_1)$ is thus the same as the distance from $y$ to $r(T_2)$, and therefore $y$ show that the property holds.
\end{proof}

Importantly, we can argue that the number of possible signatures depends only on height and $s$-boundedness.  

\begin{lemma}\label{lem:nbsigs-bound}
Let $S(s, h)$ denote the set of all possible signatures for $s$-bounded valued trees of height $h$ or less, for any layering function that map leaves to $\{0, 1, \ldots, k\}$.
Then 
\begin{align*}
    |S(s, h)| \leq \begin{cases}
        k + 1 &\mbox{if $h = 1$} \\
        (s + 2) \cdot 3^{|S(s, h - 1)|}  + |S(s, h - 1)| &\mbox{otherwise}
    \end{cases}
\end{align*}
\end{lemma}

\begin{proof}
All valued trees with $h = 1$ have a signature of the form $(\l(v))$, and $\l(v)$ can take up to $k + 1$ values.
For $h > 1$, consider an $s$-bounded valued tree of height $h$.  For each of the $S(s, h - 1)$ possible signatures of valued trees of height $h - 1$ or less, the signature vector has an entry in $\{0, 1, 2\}$, i.e. $3$ possible values for each element of $S(s, h - 1)$.
Moreover, the last entry is $\s(r(T))$, which can take up to $s + 2$ values in $\{0, 1, \ldots, s, \infty\}$.
This counts the number of signatures for a valued tree of height exactly $h$. Since $S(s, h)$ includes signatures for valued trees of height $h$ \emph{or less}, we must add the term $|S(s, h - 1)|$ to count the valued trees of height $h - 1$ or less.
\end{proof}

We note that the upper bound on $|S(s, h)|$ can be computed in time proportional to $h$ (omitting the bit manipulations required to handle the memory required to represent the large value of $|S(s, h)|$), since for each $h \geq 2$, it suffices to store $|S(s, h - 1)|$ after computing it, and use it to compute $|S(s, h)|$ in constant time, and repeat.

We now come back to $k$-leaf powers, and show that in all the cases that will be of interest to us, $s$ and $h$ are bounded by a function of $k$.

\begin{lemma}\label{lem:kbounded-and-height-general}
Assume that $G$ is a connected $k$-leaf power and let 
$X \subseteq V(G)$.  Assume that there is $C \subseteq X$ such that $C$ is a clique and $X \setminus C \subseteq N_G(C)$.
Let $T^*$ be a $k$-leaf root of $G$, and let $(T, \s)$ be the valued restriction of $T^*$ to $X$.  
Then $\T$ is $k$-bounded and $height(\T) \leq 3k$.
\end{lemma}

\begin{proof}
Let us first argue that $\T$ has height at most $3k$.
Consider the tree $T|C$.
Then $T|C$ has height at most $k$, as otherwise the members of the $C$ clique cannot be at distance at most $k$ from each other.
Moreover, any $x \in X \setminus C$ is at distance at most $k$ from some node of $T|C$, as otherwise $x$ cannot be at distance at most $k$ from any member of $C$.  Thus in $T$, $r(T)$ is at distance at most $k$ from $r(T|C)$, and any leaf in $X \setminus C$ is at distance at most $2k$ from $r(T|C)$, since the farthest possible leaf is at distance $k$ below the deepest node of $T|C$.
Hence $height(\T) \leq 3k$.

As for $k$-boundedness, 
suppose that $\T$ is not $k$-bounded.  
Then there is some $v \in V(T) \setminus L(T)$ such that $\s(v) > k$ but $\s(v) \neq \infty$.
Let 
\[L^*_v = \bigcup_{v' \in ch_{T^*}(v) \setminus ch_T(v)} L(T^*(v'))\] 
Since $(T, \s)$ is the valued restriction of $T^*$ to $X$, each leaf $w \in L^*_v$ is in $V(G) \setminus X$ and is at distance greater than $k$ from $v$.
Moreover, we know that at least one such $w$ exists, as otherwise we would have $\s(v) = \infty$.
Then $w$ only has neighbors in $L^*_v$, since any other neighbor would require a path of length at most $k$ that passes through $v$, which is not possible.  
It follows that every member of $L^*_v$ only has neighbors in $L^*_v$.  In other words, $L^*_v$ is disconnected from the rest of the graph, a contradiction since $G$ is connected.  Therefore, $\T$ is $k$-bounded.
\end{proof}

Note that the $k$ and $3k$ bounds could  be slightly improved with a more detailed analysis, but we would still obtain a power tower behavior for the number of possible signatures.

\subsection{Pruning subsets with redundant $k$-leaf root signatures}

We now establish the connections between similar structures and signatures.
Let $\S = (\C, \Y, z, \L)$ be a similar structure of $G$.  Unless mentioned otherwise, we will always assume that $\C = \{C_1, \ldots, C_d\}, \Y = \{Y_1, \ldots, Y_d\}$ and $\L = \{\l_1, \ldots, \l_d\}$ for some $d$.
We will mainly look at the $k$-leaf roots of the $G[\hC{i}]$ subgraphs.
For $i \in [d]$, let $LR_z(\hC{i})$ be the set of all $k$-leaf roots of $G[\hC{i}]$ whose root is the unique neighbor of $z$ in the tree.  Note that $LR_z(\hC{i})$ captures every possible $k$-leaf root, except that a specific node is chosen as the root.

\begin{lemma}\label{lem:kbounded-and-height}
Assume that $G$ is a connected $k$-leaf power and let 
$\S = (\C, \Y, z, \L)$ be a similar structure of $G$.
For $i \in [d]$, let $T^* \in LR_z(\hC{i})$, and let $\T = (T, \s)$
be the valued restriction of $T^*$ to $C_i \cup \{z\}$.
Then $\T$ is $k$-bounded and $height(\T) \leq 3k$.
\end{lemma}

\begin{proof}
Note that $G[\hC{i}]$ is connected, since $Y_i$ consists of connected components that have neighbors in $C_i$, and all vertices in $C_i$ have $z$ as a neighbor.
Thus, we can apply Lemma~\ref{lem:kbounded-and-height-general}, since $\{z\}$ is a clique and $C_i \subseteq N_G(z)$.
\end{proof}

Together, Lemma~\ref{lem:nbsigs-bound} and Lemma~\ref{lem:kbounded-and-height} allow us to 
bound the number of possible signatures of $k$-leaf roots of $G[\hC{i}]$ subgraphs by $|S(k, 3k)|$, which grows quickly but is a function of $k$ only.
The last piece we need is that all the $(\hC{i})$'s have $k$-leaf roots with the same signatures.

Let $s \in S(k, 3k)$ be a possible signature for a $k$-bounded valued tree of height at most $3k$, for any layering function.
We say that $\hC{i}$ \emph{accepts} $s$ if there exists
$T^* \in LR_z(\hC{i})$ such that 
$sig_{\l_i}(\T) = s$, where $\T$ is the valued restriction of $T^*$ to $C_i \cup \{z\}$.  

We then define
\[
accept(\S, C_i) = \{s \in S(k, 3k) : \hC{i} \mbox{ accepts }  s \}
\]
It is important to note that by Lemma~\ref{lem:kbounded-and-height}, for any $T^* \in LR_z(\hC{i})$, $sig_{\l_i}(\T) \in S(k, 3k)$, and therefore $sig_{\l_i}(\T)  \in accept(\S, C_i)$, where $\T$ is the valued restriction of $T^*$ to $C_i \cup \{z\}$.  In other words, $accept(\S, C_i)$ captures the signature of \emph{every} $k$-leaf root of $G[\hC{i}]$.  

We need one last definition.

\begin{definition}
A similar structure $\S = (\C, \Y, z, \L)$ of $G$ is called \emph{homogeneous}
if, for each $C_i, C_j \in \C$, $accept(C_i) = accept(C_j) \neq \emptyset$.
\end{definition}

For our purposes, we need a homogeneous similar structure of size at least $3|S(k, 3k)|$.  If all $k$-leaf roots have large enough arity, this is guaranteed to exist.

\begin{lemma}\label{lem:hom-struct}
Let $G$ be a connected $k$-leaf power. 
Let $d = 3|S(k, 3k)| \cdot 2^{|S(k, 3k)|}$, and assume that $G$ admits a $k$-leaf root of arity at least $d + 1$.
Then there is a \emph{homogeneous} similar structure $\S = (\C, \Y, z, \L)$ of $G$ such that $|\C| = 3|S(k, 3k)|$.  Moreover, for each $C_i \in \C$, $|C_i| \leq d^k$
and $G[C_i \cup Y_i \cup \{z\}]$ has maximum degree at most $d^k$.
\end{lemma}

\begin{proof}
By Lemma~\ref{lem:exists-similar}, there exists 
a similar structure $\S' = (\C', \Y', z, \L')$ of $G$ with $|\C'| = d$, with each $C'_i \in \C'$ having $|C'_i| \leq d^k$ and $G[C'_i \cup Y'_i \cup \{z\}]$ having maximum degree $d^k$ or less.  
Denote $\C' = \{C'_1, \ldots, C'_d\}$, $\Y' = \{Y'_1, \ldots, Y'_d\}$ and
$\L' = \{\l'_1, \ldots, \l'_d\}$.
The problem is that $\S'$ might not be homogeneous.

Notice that by Lemma~\ref{lem:kbounded-and-height}, for any $C'_i \in \C$, $accept(\S', C'_i)$ is a subset of $S(k, 3k)$.  
Thus the number of possible distinct $accept(\S, C'_i)$ sets is $2^{|S(k, 3k)|}$.  
By the pigeonhole principle, the fact that $|\C'| = d = 3|S(k, 3k)| 2^{|S(k, 3k)|}$ implies that there is 
some $\C \subseteq \C'$ such that $|\C| = 3|S(k, 3k)|$
and such that $accept(\S', C'_i) = accept(\S', C'_j)$ 
for every $C'_i, C'_j \in \C$.
Note that since $G$ is a $k$-leaf power, none of these $accept$ sets is empty, since each $G[\hC{i}]$ is a $k$-leaf power.

We thus know that there exists a set of indices $h_1, \ldots, h_{3|S(k, 3k)|}$ such that $\C = \{C'_{h_1}, C'_{h_2}, \ldots, C'_{h_{3|S(k, 3k)|}}\}$.  Let $\Y = \{Y'_{h_1}, Y'_{h_2}, \ldots, Y'_{h_{3|S(k, 3k)|}}\}$, and similarly let $\L = \{\l'_{h_1}, \l'_{h_2}, \ldots, \l'_{h_{3|S(k, 3k)|}}\}$.
One can easily verify that all properties of a similar structure hold for $\S = (\C, \Y, z, \L)$ (in particular, all the $C'_i \cup Y'_i$ that were not kept in $\C$ now join the same connected component as $z$).  We have $|\C| = 3|S(k, 3k)|$.
This similar structure is homogeneous.
Finally, the bounds for $|C_{h_i}| \leq d^k$
and the maximum degree at most $d^k$ for $G[C_{h_i} \cup Y_{h_i} \cup \{z\}]$ still hold for every $h_i$.
\end{proof}

We can finally present the main result of this section.

\begin{theorem}\label{thm:iffc1}
Let $G$ be a connected graph.
Let $\S = (\C, \S, z, \L)$ be a homogeneous similar structure of $G$, with $\C = \{C_1, \ldots, C_l\}, \Y = \{Y_1, \ldots, Y_l\}$ and $l = 3|S(k, 3k)|$.  
%
Then $G$ is a $k$-leaf power if and only if $G - (C_1 \cup Y_1)$ is a $k$-leaf-power.  
\end{theorem}

\begin{proof}
First, note that if $G - (C_1 \cup Y_1)$ is not a $k$-leaf-power, then by heredity, $G$ is not a $k$-leaf-power.  We now focus on the other direction of the statement.

Suppose that $G - (C_1 \cup Y_1)$ is a leaf power, and let $R$ be a $k$-leaf root of $G - (C_1 \cup Y_1)$.  Assume without loss of generality that $R$ is rooted at the single neighbor of $z$.  The proof is constructive: we show algorithmically how to insert $C_1 \cup Y_1$ into $R$ to obtain a $k$-leaf root of $G$.

For $i \in \{2, \ldots, l\}$, let $T^*_i = R|(\hC{i})$, which is a $k$-leaf root of $G[\hC{i}]$.
Furthermore let $\T_i = (T_i, \s_i)$ be the valued restriction of $T^*_i$ to $C_i \cup \{z\}$.  Note that $V(T_i) \subseteq V(T^*_i) \subseteq V(R)$.

Since $l - 1 > |S(k, 3k)|$, there must be distinct $i, j \in \{2, \ldots, l\}$ such that $sig_{\l_i}(\T_i) = sig_{\l_j}(\T_j)$, by the pigeonhole principle.
Assume, without loss of generality, that $i = 2$ and $j = 3$ (otherwise, simply rename the $C_i$'s).
Note that since $V(T_2) \subseteq V(T_2^*) \subseteq V(R)$, each node of $T_2$ and $T^*_2$ is in $R$.  The same holds for $T_3$ and $T^*_3$.
Moreover, all of $T_2, T^*_2, T_3$ and $T^*_3$ contain $z$ and hence also contain the root of $R$.  It follows that $r(T_2) = r(T^*_2) = r(T_3) = r(T^*_3) = r(R)$.

Since $accept(\S, \hC{1}) = accept(\S, \hC{2}) = accept(\S, \hC{3})$ by homogeneity, there exists a $k$-leaf root $T^*_1$ of $G[\hC{1}]$, with $r(T^*_1)$ being the parent of $z$, such that $sig_{\l_1}(\T_1) = sig_{\l_2}(\T_2) = sig_{\l_3}(\T_3)$, where $\T_1 = (T_1, \s_1)$ is the valued restriction of $T^*_1$ to $C_1 \cup \{z\}$.
Note that $r(T_1) = r(T^*_1)$, again because of $z$.
We now describe how to insert $T^*_1$ into $R$, based on the  signature of $\T_1$.
This is shown in Algorithm~\ref{alg:insertt1}.
By inserting a subtree $T^*_1(u)$ as a child of $r$, we mean to add all the nodes and edges of the $T^*_1(u)$ tree to $R$, and to add an edge between $r$ and $r(T^*_1(u))$.

\begin{algorithm2e}[H]
\SetAlgoLined
\DontPrintSemicolon
\SetKwProg{Fn}{Function}{}{end}

$insert(r(T^*_1), r(R))$ //initial call \;
\;
\Fn{insert($t, r$)}
{
    //$t \in V(T^*_1)$ is the node of $T^*_1$ we are inserting\;
    //$r \in V(R)$ is the node of $R$ we are inserting on\;
    \ForEach{child $u \in ch_{T^*_1}(t) \setminus ch_{T_1}(t)$}
    {
        Insert the $T^*_1(u)$ subtree as a child of $r$ \label{line:child-tstar}\;
    }
    
    \ForEach{child $u \in ch_{T_1}(t)$}
    {
        \uIf{$\exists w \in ch_{T_1}(t) \setminus \{u\}$ such that $sig_{\l_1}(\T_1(w)) = sig_{\l_1}(\T_1(u))$}
        {
            Insert the $T^*_1(u)$ subtree as a child of $r$ \label{line:child-two}\;
        }
        \uElse
        {
            Let $u_2 \in ch_{T_2}(r)$ such that $sig_{\l_2}(\T_2(u_2)) = sig_{\l_1}(\T_1(u))$\; 
            Let $u_3 \in ch_{T_3}(r)$ such that $sig_{\l_3}(\T_3(u_3)) = sig_{\l_1}(\T_1(u))$\;
            \uIf{$u_2 \neq u_3$}
            {
                    Insert the $T^*_1(u)$ subtree as a child of $r$ \label{line:child-onediff} \;
            }
            \uElse
            {
                \uIf{$u_2 \neq z$}
                {
                    Recursively call $insert(u, u_2)$ \label{line:recurse}\;
                }
            }
        }
    }
}
\caption{Insertion of $T_1^*$ into $R$.}
\label{alg:insertt1}
\end{algorithm2e}

Let us begin with a bit of intuition on this algorithm.
The idea is that $insert(t, r)$ embeds $T_1^*(t)$ into $R(r)$, where $t$ should be a node of $T_1$ and $r$ a node of $R$.  
The initial call  says that $r(T^*_1) = r(T_1)$ should correspond to $r(R)$ after $T^*_1$ is inserted, and the recursive calls make similar correspondences with children of $r(T^*_1)$ and $r(R)$, recursively.
We will maintain the invariant that at any point, $r \in V(R) \cap V(T_2) \cap V(T_3)$, such that $\T_2(r), \T_3(r)$ have the same signature as $\T_1(t)$ (we prove this below).  
 Then in any recursion, when $t$ has children $u$ in $T_1^*$ but not in $T_1$, we know that $T_1^*(u)$ only has leaves in $Y_1$ (line~\ref{line:child-tstar}).  In this case, we just insert the $T_1^*(u)$ subtree.  The $Y_i$ leaves to insert must be at distance more than $k$ to all of $L(R)$, and the fact that $\T_2(r)$ and $\T_3(r)$ have similar subtrees at the $r$ location helps us guarantee it.  That is, $\T_2(r)$ lets us argue on the distance relationships between $Y_1$ and the leaves in $L(R) \setminus L(\T_2(r))$, and $\T_3(r)$ helps us with all relationships with leaves in $L(\T_2(r))$.
 This idea of complementarity is the reason we need both $\T_2$ and $\T_3$.  
 A similar idea applies when $T_1^*(u)$ is inserted on lines~\ref{line:child-two} and~\ref{line:child-onediff}.  A recursion is needed for each $u \in ch_{T_1}(t)$ such that the node of $T_2$ and $T_3$ corresponding to $u$ is the same in $ch_R(r)$, since the complementarity trick cannot be applied.  
Let us proceed with the details.  

Since the algorithm only inserts subtrees as children of nodes of $R$, it is not hard to see that the modified $R$ will also be a tree.  
For the rest of the proof, let $R'$ be the tree obtained after the algorithm has terminated, assuming that the initial call is $insert(r(T^*_1), r(R))$. 
The rest of the proof is dedicated to showing that $R'$ is a $k$-leaf root of $G$.  The proof is divided in a series of claims.

\newpage

\begin{claim}\label{claim:basic}
At any point that $insert(t, r)$ is called and arguments $t$ and $r$ are received, the following holds:
\begin{itemize}
    \item 
    $t \in V(T_1) \setminus L(T_1)$;
    
    \item 
    $r \in V(T_2) \cap V(T_3)$;
    
    \item 
    $sig_{\l_1}(\T_1(t)) = sig_{\l_2}(\T_2(r)) = sig_{\l_3}(\T_3(r))$.
\end{itemize}
\end{claim}

\begin{proof}
We argue this by induction on the depth of the recursion.
As a base case, consider the initial call with $t = r(T^*_1)$ and $r = r(R)$.  We have $t \in V(T_1)$ since $r(T^*_1) = r(T_1)$.  To see that $t \notin L(T_1)$, note that $|L(T_1)| = |C_1 \cup \{z\}| \geq 2$ implies that $r(T_1)$ is not a leaf itself, since leaves have $0$ children.  Thus initially, $t \in V(T_1) \setminus L(T_1)$.
Also, we have already argued that $r = r(R) = r(T_2) = r(T_3)$.  Moreover, in this initial case, $sig_{\l_1}(\T_1(t)) = sig_{\l_1}(\T_1) = sig_{\l_2}(\T_2) = sig_{\l_2}(\T_2(r))$ and the same holds for $\T_3$.

Now, assume that the statement holds for any recursive call of depth $\delta$.
We argue that it also holds for recursive calls at depth $\delta + 1$. 
When a call at recursion depth $\delta$ with parameters $t$ and $r$ is made, the only way a deeper recursive call could be made is when line~\ref{line:recurse} is reached.  This happens when $t$ has a child $u$ in $T_1$ such that no other child of $t$ in $T_1$ has the same signature.  
Since we assume by induction that $sig_{\l_1}(\T_1(t)) = sig_{\l_2}(\T_2(r))$, by Lemma~\ref{lem:basic-sigs}.\ref{lem:basic-onechild}, $r$ must also have exactly one child $u_2$ in $T_2$ such that $sig_{\l_2}(\T_2(u_2)) = sig_{\l_1}(\T_1(u))$.  By the same argument, $r$ has exactly one child $u_3$ in $T_3$ such that $sig_{\l_3}(\T_3(u_3)) = sig_{\l_1}(\T_1(u))$.  In particular, $u_2$ and $u_3$ as used in the algorithm always exist.

A recursive call is only made if $u_2 = u_3$ and $u_2 \neq z$, in which case $u$ and $u_2$ are passed to a recursive call.
We know that $u \in V(T_1)$, and must argue that $u \notin L(T_1)$.
Assume instead that $u$ is a leaf of $T_1$.
Since $sig_{\l_1}(\T_1(u)) = sig_{\l_2}(\T_2(u_2)) = sig_{\l_3}(\T_3(u_3))$, $u_2$ would also be a leaf of $T_2$ and $u_3$ a leaf of $T_3$.  We have $L(T_2) = C_2 \cup \{z\}$ and $L(T_3) = C_3 \cup \{z\}$.  Since $C_2$ and $C_3$ are disjoint, only $u_2 = u_3 = z$ is possible, but this is a contradiction since the algorithm explicitly states that line~\ref{line:recurse} could not be reached in this case.  Hence, $u \notin L(T_1)$, which proves the first property of the claim.

The second and third properties of the claim hold for the recursive call because $u_2 = u_3$, and because they are explicitly chosen to have the same signature.
\end{proof}

We show that to create $R'$, the algorithm inserts a subtree that is leaf-isomorphic to $T^*_1$.  This ensures that distance relationships between leaves of $T^*_1$ are the same in $R'$.

\begin{claim}\label{claim:iso}
$R'|(\hC{1}) \simeq_L T^*_1$.
\end{claim}

\begin{proof}
We prove the claim by induction on the height of $T_1(t)$.  
That is, we show 
that for any $t$ and $r$ passed to the algorithm, $R'$ contains a subtree $R'_t$ that is leaf-isomorphic to $T_1^*(t)$, and that the leaf-isomorphism from $T_1^*(t)$ to $R'_t$ maps $t$ to $r$.

By Claim~\ref{claim:basic}, $t$ is in $T_1$ but not a leaf of $T_1$.  So as a base case, consider that the algorithm receives arguments $t$ and $r$, such that $T_1(t)$ is a subtree of $T_1$ of height $2$.  That is, all the children of $t$ in $T_1$ are leaves (but not necessarily all the children of $t$ in $T^*_1$).
Consider a child $u \in ch_{T^*_1}(t) \setminus ch_{T_1}(t)$.
Then line~\ref{line:child-tstar} ensures that $T_1^*(u)$ is inserted as a child subtree of $r$ in $R'$.
Consider now a child $u \in ch_{T_1}(t)$, which is a leaf by assumption. 
If some other child $w$ of $T_1$ has the same signature (i.e. $sig_{\l_1}(\T_1(u)) = sig_{\l_1}(\T_1(w))$), line~\ref{line:child-two} ensures that $u$ is inserted as a child of $r$.  
Otherwise, there are $u_2 \in ch_{T_2}(r)$ and $u_3 \in ch_{T_3}(r)$ with the same signature as $u$, i.e. $sig_{\l_2}(\T_2(u_2)) = sig_{\l_3}(\T_3(u_3)) = sig_{\l_1}(\T_1(u))$.  
As in the previous claim, $u_2$ and $u_3$ must also be leaves and one of $u_2 \neq u_3$ or $u_2 = u_3 = z$ must hold, because $C_2$ and $C_3$ are disjoint.  
If $u_2 = u_3 = z$, then $u = z$ as well, since this is the only leaf with signature $(0)$.  In this case, $t = r(T_1)$ and $r = r(R)$, and the check that $u_2 \neq z$ in the algorithm ensures that nothing happens (which is correct, since $z$ is already a child of $r$).  If $u_2 \neq u_3$, $u$ is added as a child of $r$ on line~\ref{line:child-onediff} (which is correct since $u \neq z$).
It follows that all children of $t$ in $T^*_1$ are children of $r$ in $R'$.  Therefore, $R'$ contains a subtree leaf-isomorphic to $T^*_1(t)$ with $t$ mapped to $r$.

Now assume that $T_1(t)$ is a subtree of height greater than $2$.
Notice that for each child $u \in ch_{T^*_1}(t)$, $T_1^*(u)$ is inserted as a child subtree of $r$ in $R'$, 
unless $t$ has no other child with the same signature as $u$, and there are $u_2 \in ch_{T_2}(r)$ and $u_3 \in ch_{T_3}(r)$, with $u_2 = u_3$, with the same signature as $u$.  
As before, if $u = z$, then $u_2 = u_3 = z$ and nothing happens, which is correct.
Otherwise, line~\ref{line:recurse} is reached and the algorithm is called recursively with arguments $u$ and $u_2$.
In this case, we may assume by induction that $R'$ contains a subtree leaf-isomorphic to $T^*_1(u)$, with the leaf-isomorphism mapping $u$ to $u_2$.
Let $U$ be the set of children of $t$ in $T_1$ for which line~\ref{line:recurse} is reached.  Notice that for each distinct $u, u' \in U$, $sig_{\l_1}(\T_1(u)) \neq sig_{\l_1}(\T_1(u'))$, since children with a non-unique signature are inserted on line~\ref{line:child-two}.
Moreover, each $u \in U$ is put in correspondence with a child $u_2$ of $r$ with the same signature as $u$, implying that every $u \in U$ has a distinct correspondent among the children of $r$.
We can therefore conclude  that for each $u \in ch_{T^*_1}(t)$ (other than $z$), either $T^*_1(u)$ is inserted as a child subtree of $r$, or $T^*_1(u)$ is inserted recursively with $u$ being mapped to a unique child $u_2$ or $r$.  It follows that $R'$ contains a subtree isomorphic to $T_1^*(t)$ with $t$ mapped to $r$.

To finish the proof, we observe that the above statement holds for $t = r(T_1) = r(T^*_1)$ and $r = r(R)$.  Thus $R'$ contains a subtree leaf-isomorphic to $T^*_1$ with a leaf-isomorphism that maps $r(T^*_1)$ to $r(R)$.  
\end{proof}

It remains to argue that distance relationships in $R'$ are also correct for pairs of leaves with one in $\hC{1}$ and the other in $V(G) \setminus (\hC{1}) = L(R) \setminus \{z\}$.

\begin{claim} \label{claim:cross}
Let $x \in L(T^*_1)$ and $y \in L(R) \setminus \{z\}$.
If $xy \in E(G)$, then $dist_{R'}(x, y) \leq k$, and otherwise, 
$dist_{R'}(x, y) > k$.
\end{claim}

\begin{proof}
We argue that anytime that a leaf of $T^*_1$ is inserted, the conditions of the claim are satisfied. 
Consider any recursion of the algorithm where parameters $t$ and $r$ are given, and let us focus on the set of leaves of $T^*_1$ inserted during that recursion.

Let $u \in ch_{T^*_1}(t)$ and assume that $T^*_1(u)$ is inserted as a child of $r$ in the current recursion.  Let $x \in L(T^*_1(u))$ be an inserted leaf.
We consider two possible cases.

\medskip 

\noindent
\textbf{Case 1: $x \in L(T^*_1(u)) \setminus L(T_1)$.}
Observe that $x \in L(T^*_1) \setminus L(T_1) = (\hC{1}) \setminus (C_1 \cup \{z\}) = Y_1$, and thus $x$ has no neighbors in $L(R)$.  Thus we must argue that $x$ is at distance more than $k$ to any leaf $y$ of $R$.
Let $x'$ be the lowest ancestor of $x$ in $T^*_1$ (i.e. farthest from the root) such that $x' \in V(T_1)$.  Note that if $T_1^*(u)$ was inserted on line~\ref{line:child-tstar}, then $x' = t$ since $u \notin V(T_1)$ but $t \in V(T_1)$ by Claim~\ref{claim:basic}, and otherwise, $x'$ is a descendant of $u$.
In any case, $x'$ is in $V(T_1(t))$.  Recall that $\T_1 = (T_1, \sigma_1)$, and that $\sigma_1(x')$ is the minimum distance from $x'$ to a leaf descending from a node in $ch_{T^*_1}(x') \setminus ch_{T_1}(x')$.  We note that $x$ is such a leaf, and thus $dist_{T^*_1}(x, x') \geq \s_1(x')$.

Also recall that by Claim~\ref{claim:basic}, $r \in V(T_2) \cap V(T_3)$ and $sig_{\l_1}(\T_1(t)) = sig_{\l_2}(\T_2(r)) = sig_{\l_3}(\T_3(r))$.  
By Lemma~\ref{lem:basic-sigs}.\ref{lem:basic-samex}, there is $x_2' \in V(T_2(r))$ such that $dist_{T_1}(x', t) = dist_{T_2}(x_2', r)$ and $sig_{\l_1}(\T_1(x')) = sig_{\l_2}(\T_2(x_2'))$.  
In particular, $\s_1(x') = \s_2(x_2')$.
Since $\T_2$ is the valued restriction of $T^*_2$ to $C_2 \cup \{z\}$, this implies that there exists $x_2 \in L(T_2^*) \setminus L(T_2) = Y_2$, descending from a node in $ch_{T^*_2}(x_2') \setminus ch_{T_2}(x_2')$, such that $dist_{T^*_2}(x_2, x_2') = dist_{R}(x_2, x_2') = \sigma_2(x_2') = \sigma_1(x_1') \leq dist_{T^*_1}(x, x')$.
Since the distance from $x_2'$ to $r$ is the same as the distance from $x'$ to $t$, this also implies that $dist_{R}(x_2, r) \leq dist_{T_1^*}(x, t)$.
Now, let $y \in L(R) \setminus (C_2 \cup Y_2)$. 
We have $x_2y \notin E(G)$ and obtain  
\begin{align*}
    k < dist_R(x_2, y) &\leq dist_R(x_2, r) + dist_R(r, y) \\
                 &\leq dist_{T_1^*}(x, t) + dist_{R}(r, y) \\
                 &= dist_{R'}(x, r) + dist_{R'}(r, y) \\
                 &= dist_{R'}(x, y)
\end{align*}
where, for the last two equalities, we used the fact that $T^*_1(u)$ was inserted as a child subtree of $r$, and thus that the path from $x$ to $y$ in $R'$ must first go to $r$, and then to $y$.
This shows that $x$ has the correct distance to all $y \in L(R) \setminus (C_2 \cup Y_2)$.  It remains to handle the vertices in $C_2 \cup Y_2$.
For the moment, let us instead consider $y \in L(R) \setminus (C_3 \cup Y_3)$.  We can apply the same argument as above, but using $T_3$ instead.  That is, by the equality of the signatures, there is some $x_3 \in Y_3$ such that $dist_{T^*_3}(x_3, r) = dist_R(x_3, r) \leq  dist_{T^*_1}(x, t)$.
By repeating the above argument, we deduce that for any $y \in L(R) \setminus (C_3 \cup Y_3)$, $dist_R(x_3, y) > k$, which in turn implies 
$dist_{R'}(x, y) > k$.  
In particular, this holds for any $y \in C_2 \cup Y_2$.
We have therefore shown that $dist_{R'}(x, y) > k$ for any $y \in L(R)$, and thus that $x$ is correctly inserted.

\medskip 

\noindent 
\textbf{Case 2: $x \in L(T_1)$}.  
In this case, $u \in ch_{T_1}(t)$ and $T_1^*(u)$ must have been inserted as a child of $r$ on line~\ref{line:child-two} or on line~\ref{line:child-onediff}.  Note that $x \in C_1$ ($x = z$ is not possible since $z$ is not reinserted in $R'$).

Let us define $u_2 \in ch_{T_2}(r)$ and $u_3 \in ch_{T_3}(r)$ to handle both cases.
If $T^*_1(u)$ is inserted on line~\ref{line:child-onediff}, then define $u_2$ and $u_3$ as in the algorithm.  In this case, we must have $u_2 \neq u_3$.
If instead $T^*_1(u)$ is inserted on line~\ref{line:child-two}, then there is $w \in ch_{T_1}(t)$ such that $sig_{\l_1}(\T_1(u)) = sig_{\l_1}(\T_1(w))$ and $w \neq u$.
By Lemma~\ref{lem:basic-sigs}.\ref{lem:basic-twochild}, there exist distinct $u_2, u_2' \in ch_{T_2}(r)$ and $u_3, u_3' \in ch_{T_3}(r)$ such that 
$sig_{\l_1}(\T_1(u)) = sig_{\l_2}(\T_2(u_2)) = sig_{\l_2}(\T_2(u_2')) = sig_{\l_3}(\T_3(u_3)) = sig_{\l_3}(\T_3(u_3'))$.  
Observe that at least one of $u_2 \neq u_3$ or $u_2 \neq u'_3$ holds.  Assume without loss of generality that $u_2 \neq u_3$.

In either case, note that $u_2$ and $u_3$ are defined so that $u_2 \neq u_3$, $u_2 \in ch_{T_2}(r)$, $u_3 \in ch_{T_3}(r)$, $sig_{\l_1}(\T_1(u)) = sig_{\l_2}(\T_2(u_2)) = sig_{\l_3}(\T_3(u_3))$.

By Lemma~\ref{lem:basic-sigs}.\ref{lem:basic-samex}, there exist $x_2 \in L(T_2(u_2))$
such that $dist_{T_2^*}(x_2, u_2) = dist_{T_1}(x, u)$ and $sig_{\l_1}(\T_1(x)) = sig_{\l_2}(\T_2(x_2))$.  Since the signature of a leaf is just its layer, this implies that $\l_1(x) = \l_2(x_2)$. 
Also note that $dist_{R'}(x, r) = dist_{R'}(x_2, r)$.
Similarly, there exists $x_3 \in L(T_3)$ 
such that $\l_1(x) = \l_3(x_3)$ and such that $dist_{T_3^*}(x_3, u_3) = dist_{T_1}(x, u)$, i.e. 
that $dist_{R'}(x, r) = dist_{R'}(x_3, r)$.
Also note that the paths from $x$ to $r$, $x_2$ to $r$ and $x_3$ to $r$ 
only intersect at $r$, since $u, u_2$ and $u_3$ are all distinct nodes of $R'$.

First consider $y \in L(R) \setminus (C_2 \cup Y_2 \cup C_3 \cup Y_3)$.
Since the layers are all equal, by Property~\ref{cut:layerssame} of similar structures, $xy \in E(G) \Leftrightarrow x_2y \in E(G) \Leftrightarrow x_3y \in E(G)$.
If $y$ is not a descendant of $u_2$, then the path from $x_2$ to $y$ goes through $r$, implying $dist_{R'}(x, y) = dist_{R'}(x_2, y)$.
If $y$ is not a descendant of $u_3$,
then the path from $x_3$ to $y$ goes through $r$, implying $dist_{R'}(x, y) = dist_{R'}(x_3, y)$.  
Since $u_2 \neq u_3$, $y$ cannot be a descendant of both $u_2$ and $u_3$, and so at least one of $dist_{R'}(x, y) = dist_{R'}(x_2, y)$ or $dist_{R'}(x, y) = dist_{R'}(x_3, y)$ must hold.
If $xy \in E(G)$, then $x_2y, x_3y \in E(G)$ and $dist_{R'}(x_2, y) \leq k$, $dist_{R'}(x_3, y) \leq k$ both hold, implying $dist_{R'}(x, y) \leq k$. 
Similarly, if $xy \notin E(G)$, then $x_2y, x_3y \notin E(G)$ and, in the same manner, and $dist_{R'}(x_2, y) > k$, $dist_{R'}(x_3, y) > k$ both hold, implying $dist_{R'}(x, y) > k$. 

Now consider $y \in C_2 \cup Y_2$.  Assume that $xy \notin E(G)$.
Then by Property~\ref{cut:layerssame}, $x_3y \notin E(G)$.  
If $y$ is not a descendant of $u_3$, then the path from $x_3$ to $y$ goes through $r$ and $dist_{R'}(x, y) = dist_{R'}(x_3, y) > k$.
If $y$ is a descendant of $u_3$, then we have
\[
k < dist_{R'}(x_3, y) \leq dist_{R'}(x_3, r) + dist_{R'}(r, y) = dist_{R'}(x, r) + dist_{R'}(r, y)
\]
and, since $dist_{R'}(x, r) + dist_{R'}(r, y) = dist_{R'}(x, y)$, the distance from $x$ to $y$ is correctly above $k$.

Assume that $xy \in E(G)$.  Again by Property~\ref{cut:layerssame}, $x_3y \in E(G)$. Note that $y \in Y_2$ is not possible since $Y_2$ only has neighbors in $C_2 \cup Y_2$.  So we know that $y \in C_2$, and hence $\l_2(y)$ is defined.  If $y$ is not a descendant of $u_3$, the usual argument applies, since the path from $x_3$ to $y$ goes through $r$, implying $dist_{R'}(x, y) = dist_{R'}(x_3, y) \leq k$.
Suppose that $y$ is a descendant of $u_3$.
By Property~\ref{cut:layers-gt-k} of similar structures, we must have $\l_1(x) + \l_2(y) \leq k$ (because if $\l_1(x) + \l_2(y) > k$, $xy$ could not be an edge).  As $\l_1(x) = \l_2(x_2)$, 
we have $\l_2(x_2) + \l_2(y) \leq k$ as well.  Luckily, Property~\ref{cut:layers-leq-k} is applicable to vertices in the same $C_i$, and so $x_2y \in E(G)$.
The path from $x_2$ to $y$ goes through $r$, and thus
$dist_{R'}(x, y) = dist_{R'}(x_2, y) \leq k$.
This covers the case $y \in C_2 \cup Y_2$.

The last remaining case is $y \in C_3 \cup Y_3$.  This can be handled exactly as the previous case, but swapping the roles of $x_2$ and $x_3$.  
This completes the proof, as we have covered every possible case for a leaf inserted by the algorithm, at any point.
\end{proof}

To conclude the proof, we can argue that $R'$ is a $k$-leaf-power of $G$ by looking at every pair of vertices of $G$.
First note that by Claim~\ref{claim:iso}, all the leaves of $T^*_1$ are present in $R'$, and so $V(G) = L(R')$.
Let $x, y \in V(G)$, with $x \neq y$.
If $x, y \in V(G) \setminus (C_1 \cup Y_1)$, then 
$x, y \in L(R)$ and, since $dist_{R}(x, y) = dist_{R'}(x, y)$, we know that $xy \in E(G) \Leftrightarrow dist_{R'}(x, y) \leq k$. 
If $x, y \in C_1 \cup Y_1 \cup \{z\}$, then by Claim~\ref{claim:iso}, 
$R'|(C_1 \cup Y_1 \cup \{z\})$ is leaf-isomorphic to $T^*_1$.  Since $T^*_1$ is a $k$-leaf root of $G[\hC{1}]$, 
we know that $xy \in E(G) \Leftrightarrow dist_{R'}(x, y) \leq k$.
Finally, if $x \in C_1 \cup Y_1$ and $y \in V(G) \setminus (C_1 \cup Y_1)$, Claim~\ref{claim:cross} ensures that $xy \in E(G) \Leftrightarrow dist_{R'}(x, y) \leq k$.
Therefore, $R'$ is a $k$-leaf-power of $G$.
\end{proof}

\section{Computing the set of accepted signatures} \label{sec:tw}

We have not yet discussed how to find an homogeneous similar structure $\S$.
Since $k$ is fixed, it is not too hard to find a similar structure of $G$ with the properties of Lemma~\ref{lem:hom-struct}, if one exists.
It suffices to brute-force every combination of $3|S(k, 3k)|$ subsets of $V(G)$ of size at most $d^k$ in $G$ to find the $C_i$'s, and check that all similar structure properties hold.
In particular, each $G[C_i \cup Y_i \cup \{z\}]$ has maximum degree at most $d^k$, so using Lemma~\ref{lem:prelim:bounded-arity}, we can check whether this is a $k$-leaf power.
However, this is not enough, since homogeneity requires enumerating all accepted signatures for the found $\hC{i}$, in order to construct $accept(\S, \hC{i})$ and ensure that they are all equal.
We will achieve this through a tree decomposition of $G[\hC{i}]$.

Let us assume that $\S = (\C, \Y, z, \L)$ is a similar structure satisfying $|\C| = 3|S(k, 3k)|$, with each $|C_i| \leq d^k$ such that $G[\hC{i}]$ has maximum degree at most $d^k$.
We want to compute $accept(\S, C_i)$ for each $i \in [d]$.
This can be achieved using a slightly more general result.

\begin{lemma}\label{lem:enumerate}
Let $G$ be a connected graph of maximum degree at most $d^k$ and let $z \in V(G)$.  Then in time $O(n \cdot d^{8k} \cdot f(k)^3)$, where $n = |V(G)|$ and $f(k) = d^{4kd^{4k}} \cdot (k + 2)^{d^{4k}}$, one can enumerate the set of all valued trees $\mathbb{T} = \{\T_1, \ldots, \T_l\}$, up to value-isomorphism, such that $\T_i \in \mathbb{T}$ 
if and only if there exists a $k$-leaf root $T^*$ of $G$ such that (1) $T^*$ is rooted at the parent of $z$; and (2) $\T_i$ is the valued restriction of $T^*$ to $N_G[z]$.
\end{lemma}

To see why Lemma~\ref{lem:enumerate} allows us to compute $accept(\S, C_i)$, recall that the latter contains the signature of every $\T$ that is the valued restriction of $T^*$ to $C_i \cup \{z\}$, where $T^*$ is a $k$-leaf root of $G[\hC{i}]$ with the root as the parent of $z$.  Since $N_{G[\hC{i}]}[z] = C_i \cup \{z\}$, we can pass $G[\hC{i}]$ and $z$ to the above lemma.
By taking the signature of every $\T_i$ returned, we obtain $accept(\S, C_i)$.  Note that the lemma does not deal with layers, so the leaf sets of the desired $\T_i$'s are $N_G[z]$ (instead of integers representing layers as in the previous section).

The rest of this section is dedicated to the proof of Lemma~\ref{lem:enumerate}.
We will write $N(v)$ and $N[v]$ instead of $N_G(v)$ and $N_G[v]$, respectively, with the understanding that $G$ is the graph stated in Lemma~\ref{lem:enumerate}.  Likewise, for $X \subseteq V(G)$, we write $N(X)$ and $N[X]$ instead of $N_G(X)$ and $N_G[X]$, respectively.
The proof is based on the tree decomposition of $G$ and uses a relatively standard dynamic programming algorithm.
Recall that given a graph $G$, a nice tree decomposition of $G$ is a tree $B = (V_B, E_B)$ in which (1) $B_i \subseteq V(G)$ for each $B_i \in V_B$; (2) for each $uv \in E(G)$, there is some $B_i \in V_B$ with $u, v \in B_i$; (3) for each $u \in V(G)$, the set of $B_i$'s that contain $u$ form a connected subgraph of $B$.
Moreover, each $B_i \in V_B$ can be one of four types: $B_i \in L(B)$, in which case $B_i = \{v\}$ for some $v \in V(G)$; $B_i$ is an introduce node, in which case $B_i$ has a single child $B_j$ with $B_j = B_i \setminus \{v\}$ for some $v \in V(G)$; $B_i$ is a forget node, in which case $B_i$ has a single child $B_j$ with $B_i = B_j \setminus \{v\}$ for some $v \in B_j$; $B_i$ is a join node, in which case $B_i$ has two children $B_l, B_r$ and $B_i = B_l = B_r$.
The width of $B$ is $max_{B_i \in V_B}(|B_i| - 1)$, and the treewidth of $G$ is the minimum width of a nice tree decomposition of $G$.

We note that~\cite{eppstein2020parameterized} also use a tree decomposition based algorithm.  However, it does not seem adaptable directly for our purposes, since it is not guaranteed to return every structure of every $k$-leaf root of the given graph, and it does not maintain the $\s$ information that we need.

For our algorithm, 
we first check whether $G$ is chordal: if not, we know by Lemma~\ref{lem:chordal} that $G$ is not a $k$-leaf power and we can return $\mathbb{T} = \emptyset$.
Otherwise, since $G$ has maximum degree at most $d^k$ and is chordal, $G$ has clique number at most $d^k + 1$ and thus treewidth at most $d^k$.  
The properties of the tree decomposition that we will need are summarized here.

\begin{lemma}\label{lem:treedecomp}
Let $G$ be a connected chordal graph of maximum degree at most $d^k$, and let $z \in V(G)$.  Then there exists a nice tree decomposition $B = (V_B, E_B)$ of $G$ with $O(|V(G)|)$ nodes and of width at most $d^k$, such that $r(B) = \{z\}$ and such that each $B_i \in V_B$ is a non-empty clique.
\end{lemma}

We omit the proof, which uses standard arguments.  The idea is that connected chordal graphs admit a tree decomposition in which every bag is a non-empty clique (see e.g.~\cite{blair1993introduction}).  We can take such a decomposition, root it at a bag containing $z$, and apply the standard transformation to make it nice (if the root is not exactly $\{z\}$, it can be made $\{z\}$ by adding enough forget nodes above).  

Let $B = (V_B, E_B)$ be a nice tree decomposition that satisfies all the properties of Lemma~\ref{lem:treedecomp}.
For $B_i \in V_B$, we will denote by $V_i = \bigcup_{B_j \in B(B_i)} B_j$ the set of vertices of $G$ found in the bags at $B_i$ or its descendants.
Note that for each $V_i$, $G[V_i]$ is connected.  This can be seen inductively.  First, it is true for leaves.  At introduce nodes, we introduce a vertex $v$ connected to the rest of $B_i$, so $G[V_i]$ remains connected.  At forget nodes, $V_i$ is the same as $V_j$ and remains connected.  At join nodes, $V_i$ is the union of the vertices of two connected graphs that intersect at $B_i$, and thus remains connected.  

Let $B_i$ be a bag of $B$ and let $(T, \s)$ be a valued tree.
We say that $(T, \s)$ is \emph{valid for} $B_i$ if it satisfies the following conditions:

\begin{itemize}
    \item 
    $L(T) = N[B_i] \cap V_i$;
    
    \item 
    there exists a $k$-leaf root $T^*$ of $G[V_i]$ such that $(T, \s)$ is the valued restriction of $T^*$ to $N[B_i] \cap V_i$, and such that $r(T) = r(T^*)$.  
\end{itemize}

Since $r(B) = \{z\}$, our goal is to obtain the set \emph{all} of valued trees that are valid for $r(B)$, as this will yield the valued restrictions to $N[z]$ required by Lemma~\ref{lem:enumerate}.  
Note that unlike in the previous section, there is no requirement on the root of $T$ or $T^*$ being the parent of $z$.  
The requirement that $r(T) = r(T^*)$ is there to ensure that the root we see in the restricted $T$ is also the root in $T^*$ (since otherwise, the restriction could hide nodes of $T^*$ above the root of $T$).  This is important for our purposes, since it allows us to safely look at the valid valued trees for $r(B) = \{z\}$ whose root is the parent of $z$.  
Also note that it is tempting to define $(T, \s)$ as the valued restriction of $T^*$ to $B_i$, and not bother with $N[B_i] \cap V_i$.  This does not quite work --- the inclusion of the neighborhood is necessary to retain enough information to update the $\s(w)$ values accurately (see proof for details).

For the rest of this section, we shall slightly abuse notation and treat two valued trees that are value-isomorphic as \emph{the same}.  In other words, we assume the understanding that two valued trees are \emph{distinct} only if they are not value-isomorphic.

We first show that the number of distinct $(T, \s)$ valued trees to consider is bounded (with a crude estimate), meaning that we can enumerate all candidates. 

\begin{lemma}\label{lem:valid-bounded}
Let $B_i \in V(B)$ and let $(T, \s)$ be a valid valued tree for $B_i$.
Then $T$ has at most $d^{4k}$ nodes and is $k$-bounded. 
Consequently, there are at most $d^{4kd^{4k}} \cdot (k + 2)^{d^{4k}}$ possible valid valued trees for $B_i$ (up to value-isomorphism).
\end{lemma}

\begin{proof}
Let $T$ be a valid valued tree for $B_i$.  Then $L(T) = N[B_i] \cap V_i$.  First recall that $|B_i| \leq d^k + 1$.
Since $G$ has maximum degree at most $d^k$, $N[B_i] \leq d^k + 1 + (d^k + 1) \cdot d^k \leq 2d^{2k}$, assuming $d \geq 2$.  Let $T'$ be the unique tree such that $T'$ has no node with only one child, and such that $T$ is a subdivision of $T'$.  In other words, $T'$ is obtained from $T$ by contracting every node of degree $2$ (except the root).  Then $T'$ has at most $|L(T')| \leq 2d^{2k}$ internal nodes and at most $4d^{2k}$ edges.
Moreover, each edge $uv$ of $T'$ represents a path $(u, x_1, \ldots, x_l, v)$ in $T$, where $u, x_1, \ldots, x_l$ have one child.  
Note that $l \leq k$, as otherwise, the leaves in $T(v)$ would be at distance more than $k$ from the other leaves, whereas $N[B_i] \cap V_i$ is a connected subgraph of $G$.
We may thus assume that each edge of $T'$ corresponds to at most $k$ additional nodes in $T$.  Therefore, $T$ has at most $2d^{2k}$ leaves and at most $4kd^{2k}$ internal nodes, for a total of at most 
$(4k + 2)d^{2k} \leq d^{4k}$ nodes (assuming $d \geq 2$).

The number of trees with at most $d^{4k}$ nodes is bounded by $(d^{4k})^{d^{4k}} = d^{4kd^{4k}}$.
We know that $(T, \s)$ is $k$-bounded by Lemma~\ref{lem:kbounded-and-height-general} (since $N[B_i] \cap V_i$ consists of a clique $B_i$ and a subset of $N[B_i]$).
Thus, for each possible tree, each of the at most $d^{4k}$ internal nodes can receive a value in $\{0, 1, \ldots, k, \infty\}$, i.e. $k + 2$ possibilities.  Thus the number of valued trees is bounded by $d^{4kd^{4k}} \cdot (k + 2)^{d^{4k}}$.
\end{proof}

We now describe a dynamic programming recurrence over $B$ that constructs a set $Q[B_i]$ for each $B_i$.  The set $Q[B_i]$ must contain \emph{every} valid valued tree for $B_i$, i.e. $\T \in Q[B_i]$ if and only if $\T$ is valid for $B_i$.  We assume that we are enumerating each candidate valued tree from the above lemma, and must decide whether to put it in $Q[B_i]$ or not.
For convenience, denote $\N[B_i] := N[B_i] \cap V_i$ for the rest of this section.

\begin{itemize}
    \item 
    \emph{Leaf node}.   Let $B_i = \{v\}$ be a leaf of $B$.  Then $(T, \s) \in Q[B_i]$ if and only if $T$ is the single node $v$ (and $\s$ has an empty domain).

    \item 
    \emph{Introduce node}.  Let $B_i$ be an introduce node with child $B_j = B_i \setminus \{v\}$.
    Then put $(T, \s) \in Q[B_i]$ if and only if there exists $(T_j, \s_j) \in Q[B_j]$ such that all the following conditions are satisfied:
    \begin{itemize}
        \item 
        $T|\N[B_j] \simeq_L T_j$.  Let $\mu$ be the leaf-isomorphism from $T|\N[B_j]$ to $T_j$;

        \item 
        for every internal node $w \in V(T|\N[B_j])$, $\s(w) = \s_j(\mu(w))$ and for every internal node $w \in V(T) \setminus V(T|\N[B_j])$, $\s(w) = \infty$;
        
        \item 
        for each $w \in L(T) \setminus \{v\}$, 
        $vw \in E(G)$ if and only if $dist_T(v, w) \leq k$;
        
        \item 
        for each internal node $w \in V(T)$, $dist_T(v, w) + \s(w) > k$.

    \end{itemize}

    \item 
    \emph{Forget node}.  Let $B_i$ be a forget node with child $B_j = B_i \cup \{v\}$.
    Then put $(T, \s) \in Q[B_i]$ if and only if there exists $(T_j, \s_j) \in Q[B_j]$ such that
    \begin{itemize}
        \item 
        $r(T_j)$ has at least two distinct children $u, v$ such that $T_j(u)$ and $T_j(v)$ both have a leaf in $\N[B_i]$;

        \item 
        $T \simeq_L T_j|\N[B_i]$.
        Let $\mu$ be the leaf-isomorphism from $T$ to $T_j|\N[B_i]$;  
        
        \item 
       for each $w \in V(T) \setminus L(T)$, let $C(w) = \{c \in ch_{T_j}(\mu(w)) : \mu^{-1}(c) = \emptyset\}$, and let 
       \begin{align*}
           \hat{L} = \bigcup_{c \in C(w)} L(T_j(c)) \quad \mbox { and } \quad \hat{I} = \{w\} \cup \bigcup_{c \in C(w)} V(T_j(c)) \setminus \hat{L}
       \end{align*}
       Then 
       \begin{align*}
          \s(w) = \min \begin{cases}
            \min_{v \in \hat{I}} \s_j(\mu(v)) + dist_{T_j}(\mu(w), \mu(v)) \\
            \min_{l \in \hat{L}} dist_{T_j}(\mu(w), l)
          \end{cases}
       \end{align*}
    \end{itemize}

    \item 
    \emph{Join node}.  Let $B_i$ be a join node of $B$ with children $B_l$ and $B_r$, where $B_i = B_l = B_r$.  
    Then $(T, \s) \in Q[B_i]$ if and only if there exists $(T_l, \s_l) \in Q[B_l]$ and $(T_r, \s_r) \in Q[B_r]$ that satisfy:
    
    \begin{itemize}
        \item 
        $T|\N[B_l] \simeq_L T_l$.  Let $\mu_l$ be the leaf-isomorphism from $T|\N[B_l]$ to $T_l$; 
        
        \item 
        $T|\N[B_r] \simeq_L T_r$.  Let $\mu_r$ be the leaf-isomorphism from $T|\N[B_r]$ to $T_r$;

        \item 
        for each $w \in V(T) \setminus L(T)$, we have $\s(w) = \min(\s_l(\mu_l(w)), \s_r(\mu_r(w)))$ if $w$ is in both $T|\N[B_l]$ and $T|\N[B_r]$, 
        $\s(w) = \s_l(\mu_l(w))$ if $w$ is only in $T|\N[B_l]$, and $\s(w) = \s_r(\mu_r(w))$ otherwise;
        
        \item 
        for each distinct $u, v \in L(T)$, $uv \in E(G)$ if and only if $dist_T(u, v) \leq k$;
        
        \item 
        for each $u \in \N[B_l] \setminus B_i$ and each internal node $w \in V(T|\N[B_r])$, $dist_{T}(u, w) + \s_r(\mu_r(w)) > k$.
        Likewise, for each $u \in \N[B_r] \setminus B_i$ and each $w \in V(T|\N[B_l])$, $dist_{T}(u, w) + \s_l(\mu_l(w)) > k$;
        
        \item 
        for each internal node $w_l \in V(T|\N[B_l])$ and each internal node $w_r \in V(T|\N[B_r])$, $\s_l(\mu_l(w_l)) + dist_T(w_l, w_r) + \s_r(\mu_r(w_r)) > k$.
    \end{itemize}

\end{itemize}

\begin{lemma}\label{lem:q-is-correct}
Consider the recurrence given above.  For any bag $B_i$, $(T, \s)$ is a valid valued restriction for $B_i$ if and only if $(T, \s) \in Q[B_i]$ (up to value-isomorphism).
\end{lemma}

\begin{proof}
The proof is by induction on the depth of $B_i$.
As a base case, assume that $B_i$ is a leaf.
Then $B_i = \{v\}$, $\N[B_i] = \{v\}$, and $(T, \s)$ with $T$ containing $v$ only is the only possible valid valued tree.  Hence the leaf case is correct.

Now consider the induction step.  Let $B_i \in V(B) \setminus L(B)$ and assume that the statement is correct for all descendants of $B_i$.
We prove the two directions of the statement separately.

\medskip 

\noindent
($\Rightarrow$) : suppose that $(T, \s)$ is a valid valued restriction for $B_i$.  Then there exists a $k$-leaf-power $T^*$ of $G[V_i]$ such that $(T, \s)$ is the valued restriction of $T^*$ to $\N[B_i]$, and such that $r(T) = r(T^*)$.
We must argue that $(T, \s)$ is in $Q[B_i]$.

\medskip 

\noindent 
\paragraph{Introduce node.}
Suppose that $B_i$ is an introduce node, with child $B_j = B_i \setminus \{v\}$.  Notice that by the properties of tree decompositions, $N(v) \cap V_i = B_i \setminus \{v\}$ (since $B_i$ is a clique).

Let $T^*_j = T^*|V_j$, noting that $V_j = V_i \setminus \{v\}$. 
Moreover, $T^*_j$ is a $k$-leaf root of $G[V_j]$, since distances between leaves do not change when taking a restriction.
Let $(T_j, \s_j)$ be the valued restriction of 
$T^*_j$ to $\N[B_j]$. 

We would like to use induction and argue that $(T_j, \s_j) \in Q[B_j]$, but for that we need $r(T_j) = r(T^*_j)$, which is not immediate.
Indeed, it is possible that $r(T_j) \neq r(T^*_j)$, which happens if and only if in $T^*_j$, some ancestor of $r(T_j)$ (other than itself) has descending leaves of $V_j$ not in $\N[B_j]$.  We argue that this does not occur.
Let $L' = V_j \setminus L(T^*_j(r(T_j)))$, and assume that $L'$ is non-empty.  
For that to happen, observe that $v \notin L(T^*(r(T_j)))$.  This is because $r(T) = r(T^*)$, and if $v$ was a descendant of $r(T_j)$ in $T^*$, then the root of $T = T^*|\N[B_i]$ would still be $r(T_j) \neq r(T^*)$.  Thus in $T^*$, the path from $r(T_j)$ to $v$ goes ``above" $r(T_j)$.
Also note that no member of $L'$ has a neighbor in $B_j$, as otherwise this member of $L'$ would be included in $T_j = T^*_j|\N[B_j]$.  Moreover, members of $L'$ are not neighbors of $v$.
Let $u \in L'$ be at minimum distance from $r(T_j)$ in $T^*_j$.
Then $u$ is also at minimum distance from $r(T_j)$ in $T^*$, among all leaves in $L'$.
Let $d_{u} = dist_{T^*}(u, r(T_j))$ and $d_v = dist_{T^*}(v, r(T_j))$.
First assume that $d_u \leq d_v$.
We know that $B_j \neq \emptyset$ and that $B_i$ is a clique, implying that there is $v' \in B_j$ such that $vv' \in E(G)$ and that $dist_{T^*}(v, v') \leq k$.  Since $v' \in B_j$, $v'$ descends from $r(T_j)$ and thus the path from $v$ to $v'$ passes through $r(T_j)$ in $T^*$.  Then $d_u \leq d_v$ implies $dist_{T^*}(u, v') \leq k$.  This is not possible since $u$ has no neighbor in $B_j$.
Suppose instead that $d_u > d_v$.  Note that $u$ must have at least one neighbor $u'$ in $L(T^*(r(T_j)))$ (otherwise, since $u$ is at minimum distance from $r(T_j)$, this would imply that no member of $L'$ has a neighbor outside of $L'$, and thus contradict that $G[V_i]$ is connected).  Note that $u' \notin B_j$ because $u$ has no neighbor in $B_j$.  The path from $u$ to $u'$ in $T^*$  goes through $r(T_j)$, and $u'$ is at distance at most $k$ from $u$.  Since $d_v < d_u$, $u'v \in E(G)$, a contradiction since all the neighbors of $v$ are in $B_j$.  
We deduce that $L'$ is empty.
Hence, we can safely assume that $r(T^*_j) = r(T_j)$.
Therefore, by induction, $(T_j, \s_j) \in Q[B_j]$. 
It remains to show that $(T_j, \s_j)$ satisfies all conditions of introduce nodes. 

Let us argue that $T|\N[B_j] = T_j$.
Note that $V(T|\N[B_j])$ is the set of all vertices of $T^*$ on the path between two leaves in $\N[B_j]$, and $V(T_j)$ the set of all vertices of $T^*_j$ between two leaves in $\N[B_j]$.  This makes it clear that $T|\N[B_j] = T_j$, but let us argue this for completeness.
Consider two vertices $x, y \in \N[B_j]$, and let $P_{xy}$ be the path of $T^*$ from $x$ to $y$.  
All vertices of $P_{xy}$ must be in $T$ since $x, y \in \N[B_i]$, and all vertices of $P_{xy}$ are in $T|\N[B_j]$ since $x, y \in \N[B_j]$. 
All vertices of $P_{xy}$ must be in $T^*_j$ since $x, y \in V_j$, and must also all be in $T_j$ since $x, y \in \N[B_j]$.  Since this holds for every $x, y \in \N[B_j]$, we obtain $T|\N[B_j] = T_j$.
We also have the leaf-isomorphism $\mu(w) = w$ for all $w \in V(T|\N[B_j])$ since $T|\N[B_j]$ and $T_j$ share the same set of vertices of $T^*$. 
We shall use $w$ and $\mu(w)$ interchangeably, since they represent the same node.

We next justify $\s(w) = \infty$ for every internal node $w \in V(T) \setminus V(T_j)$.  
There are two possible cases.
First suppose that $r(T) \neq r(T_j)$.  Recall that $T$ and $T_j$ differ only by the $v$ leaf.  Thus, $r(T) \neq r(T_j)$ happens if and only if $r(T^*)$ has at least two distinct two children, such that one leads to $v$ and the other leads to $r(T_j)$.  
In this situation, $w \in V(T) \setminus V(T_j)$ must be a node on one of these two paths (excluding $r(T_j)$ and $v$).  
As we have already argued, all leaves of $T^*_j$ are descendants of $r(T_j)$ in $T^*$.  Therefore, $w$ cannot have a descendant from $T^*_i$ that is not in $V_i \setminus \N[B_i]$.  If follows that $\s(w) = \infty$ represents the correct distance information for $w$.

The second case is when $r(T) = r(T_j)$.
In this case, the only difference between $T_j$ and $T$ is that there is some vertex $x \in V(T_j)$ and a path $P = (x_1, \ldots, x_l = v)$ such that $x_1 \in ch_{T}(x)$ and $x_1 \notin V(T_j)$.  In other words, $T$ was obtained from $T_j$ by appending the path leading to $v$ under some node $x$.  It follows that only the nodes $x_1, \ldots, x_{l-1}$ are internal nodes in $T$ but not in $T_j$.
The recurrence assumes that $\s(x_i) = \infty$ for each $x_i$.
Again, we must argue that in $T^*$, no such $x_i$ has a descending leaf in $L(T^*) \setminus \N[B_i]$, i.e. that $L(T^*(x_1)) = \{v\}$.
Assume that this is not the case.
Let $u = \neq v$ be a leaf of $T^*(x_1)$ at minimum distance from $x$.  This is similar to the previous situation.  We have that $u \notin \N[B_j]$ since $x_1$ is not in $T_j$.
Let $d_{u} = dist_{T^*}(u, x)$ and $d_{v} = dist_{T^*}(v, x)$.
Assume that $d_{u} \leq d_{v}$.  
Then the path from $v$ to any $v' \in B_j$ (which exists) passes through $x$, implying that $uv' \in E(G)$, a contradiction.
Assume instead that $d_u > d_v$.  
There must exist $u' \in N(u)$ such that $u' \notin L(T^*(x_1))$ (otherwise, no leaf in $L(T^*(x_1))$ has a neighbor in $G$ outside of $L(T^*(x_1))$ and $G[V_i]$ is not connected).  Note that $u' \notin B_j$.  The path from $u$ to $u'$ goes through $x$, implying that $vu' \in E(G)$, a contradiction.  Thus in $T^*$, each $x_i$ only has $v$ has a leaf descendant.  This justifies putting $\s(x_i) = \infty$ for each $x_i$.

Let $w \in V(T|\N[B_j])$ be an internal node that is in both $T$ and $T_j$.
Let $w' \in ch_{T^*}(w) \setminus ch_{T}(w)$.  Then $w'$ only leads to leaves in $V_i \setminus \N[B_i] = V_j \setminus \N[B_j]$ and must also be in $ch_{T^*_j}(w) \setminus ch_{T_j}(w)$.  
Conversely, let $w'' \in ch_{T^*_j}(w) \setminus ch_{T_j}(w)$.  Then in $T^*_j$, $w''$ only leads to leaves in $V_j \setminus \N[B_j]$.
We have a problem if $w'' \in ch_{T}(w)$.  Again, since $T$ and $T_j$ differ only by $v$, this is only possible if $w''$ has $v$ as a descendant.  But as we argued in the previous paragraph, $w''$ would not have descendants in $V_j \setminus \N[B_j]$ (here, $w''$ plays the same role as $x_1$).  Therefore, we may assume that $w'' \in ch_{T^*}(w) \setminus ch_{T}(w)$.
Since the set of hidden children of $w$ is the same in either $T$ and $T_j$, we deduce that $\s(w) = \s_j(w)$ represents the correct distance information.

Also, since $T = T^*|\N[B_i]$ and $T^*$ is a $k$-leaf root of $G$, it is clear that for all $w \in L(T) \setminus \{v\}$, $vw \in E(G) \Leftrightarrow dist_T(v, w) \leq k$.  Moreover, since we have argued that $\s(w)$ is correct for each $w \in V(T) \setminus L(T)$, we must have $dist_T(v, w) + \s(w) > k$, as otherwise $v$ would have a neighbor outside of $\N[B_i]$.  This agrees with the recurrence, and it will therefore put $(T, \s)$ in $Q[B_i]$.

\medskip 

\noindent 
\paragraph{Forget node.}
Suppose that $B_i$ is a forget node, with child $B_j = B_i \cup \{v\}$.  Note that $T^*$ is a $k$-leaf root of $G[V_j]$.
Let $(T_j, \s_j)$ be the valued restriction of $T^*$ to $\N[B_j]$. 
Because $\N[B_i] \subseteq \N[B_j]$, the root of $T_j$ satisfies $r(T_j) = r(T) = r(T^*)$.
Then by induction, $(T_j, \s_j) \in Q[B_j]$.

Note that $r(T)$ must have two distinct children with descendants in $\N[B_i]$, and thus the same holds for $T_j$, as in the recurrence.
Because $T = T^*|\N[B_i]$ and $T_j = T^*|\N[B_j]$, it is not hard to see that $T = T_j|\N[B_i]$.
The leaf-isomorphism is $\mu(w) = w$ for all $w \in V(T)$ because $T$ and $T_j|\N[B_i]$ use the same set of vertices.

Now let $w \in V(T) \setminus L(T)$.  Let $L^* = \bigcup_{c \in ch_{T^*}(w) \setminus ch_T(w)} L(T^*(c))$.
To find the minimum distance from $w$ to a leaf $u \in L^*$, we must first consider all leaves of $L(T_j)$ that descend from a child $c \in ch_{T_j}(w) \setminus ch_T(w)$, in which case $\mu^{-1}(c) = \emptyset$.  The minimum distance to such a leaf is given by $\min_{l \in \hat{L}} dist_{T_j}(\mu(w), l)$ in the recurrence.  We must also consider all leaves of $V_j \setminus L(T_j)$ that descend from a child $c \in ch_{T^*}(w) \setminus ch_{T}(w)$. 
Consider such a leaf $u$ at minimum distance from $w$, and let $w'$ be the lowest ancestor of $u$ that is in $V(T_j)$ (which exists, since $w$ is an ancestor of $u$).  
Then either $w' = w$ if $u$ descends from some child $c \in ch_{T^*}(w) \setminus ch_{T_j}(w)$, in which case the distance is $\s_j(w)$, or $w'$ is a descendant of some child of $w$ in $ch_{T_j}(w) \setminus ch_T(w)$, in which case this distance is $\s_j(w') + dist_{T_j}(w, w')$.  
In any case, the distance to such a $u$ is given by the expression $\min_{v \in \hat{I}} \s_j(\mu(v)) + dist_{T_j}(\mu(v), \mu(w))$ in the recurrence.
Since the recurrence takes the minimum over all possibilities, it will take the minimum possible distance, and thus $\s(w)$ describes the same distance as in the recurrence.

\medskip 

\noindent 
\paragraph{Join node.}
Suppose that $B_i$ is a join node, with children $B_l = B_r = B_i$.
In this part of the proof, we will use $B_i, B_l$ and $B_r$ interchangeably, but use $B_l$ and $B_r$ when we want to emphasize that we are dealing with the ``left'' or ``right'' side of the decomposition.
Let $T^*_l = T^*|V_l$ and $T^*_r = T^*|V_r$.
Let $(T_l, \s_l)$ be the valued restriction of $T^*_l$ to $\N[B_l]$.  
Similarly, let $(T_r, \s_r)$ be the valued restriction of $T^*_r$ to $\N[B_r]$.  

As in the introduce case, we would like to argue that $r(T_l) = r(T^*_l)$.  Assume that this is not the case, and that $r(T_l) \neq r(T^*_l)$.  
Then $L(T^*(r(T^*_l))) \setminus L(T^*(r(T_l)))$ is non-empty.  Let $u \in V_l$ be one of those leaves, and choose $u$ that has minimum distance to $r(T_l)$.  Notice that $u \notin \N[B_l]$, and thus $u \in V_l \setminus \N[B_l]$.  Moreover, since $r(T^*) = r(T)$, $r(T^*)$ has at least two distinct children with a descendant in $\N[B_i]$.  Let $w$ be such a child, chosen so that $w$ does not have $r(T_l)$ as a descendant.  Let $v \in L(T^*(w)) \cap \N[B_i]$.  Note that $v \notin \N[B_l]$ by the choice of $w$.  Thus, $v \in \N[B_i] \setminus \N[B_l] = N(B_r) \cap V_r$.  In particular, $v \in V_r \setminus B_i$, and 
it is not hard to see that by the properties of tree decompositions, this implies $u \neq v$.

Now, let $d_v = dist_{T^*}(v, r(T_l))$ and $d_u = dist_{T^*}(u, r(T_l))$.  Assume that $d_u \leq d_v$.  We know that $v$ has a neighbor $v'$ in $B_i$, and that all of $B_i$ descend from $r(T_l)$ in $T^*$.  Thus the path from $v$ to $v'$ passes through $r(T_l)$.  Since $dist_{T^*}(v, v') \leq k$, the assumption $d_u \leq d_v$ means that $dist_{T^*}(u, v') \leq k$, and thus $uv' \in E(G)$.  This is a contradiction since $u \notin \N[B_i]$.
Assume that $d_u > d_v$.  One can see that $u$ must have a neighbor $u'$ in $L(T^*(r(T_l)))$ (if not, by the choice of $u$, all the leaves in $L(T^*(r(T^*_l))) \setminus L(T(r(T_l)))$ have only neighbors in $L(T^*(r(T^*_l))) \setminus L(T(r(T_l)))$, contradicting that $G[V_i]$ is connected).  Recall that $u \in V_l \setminus \N[B_l]$, and so $u'$ must be in $V_l \setminus B_l$.  Moreover, $d_u > d_v$ implies that $vu' \in E(G)$.  But $v \in V_r \setminus B_i$ and $u' \in V_l \setminus B_l$, which is not allowed by the properties of tree decompositions.  

We deduce that $r(T_l) = r(T^*_l)$.
By a symmetric argument, $r(T_r) = r(T^*_r)$.  
By induction, $(T_l, \s_l) \in Q[B_l]$ and $(T_r, \s_r) \in Q[B_r]$.

Since $T = T^*|\N[B_i]$ and $\N[B_l] \subseteq \N[B_i]$, 
it is not hard to see that $T|\N[B_l] = T_l$, with the leaf-isomorphism $\mu_l(w) = w$ for $w \in V(T|\N[B_l])$.
Likewise, $T|\N[B_r] = T_r$, with the leaf-isomorphism $\mu_r(w) = w$ for $w \in V(T|\N[B_r])$.  We shall not distinguish $w$ and $\mu_l(w), \mu_r(w)$ since they refer to the same node.

Let $w \in V(T) \setminus L(T)$.  Note that $w$ is in at least one of $T|\N[B_l]$ or $T|\N[B_r]$ since $w$ has a descendant in at least one of $\N[B_l]$ or $\N[B_r]$.
If $w$ is in both $T|\N[B_l]$ and $T|\N[B_r]$, then $\s(w) = \min(\s_l(w), \s_r(w))$ holds because all leaves of $V_i \setminus \N[B_i]$ that have $w$ as their first ancestor in $T$ are either leaves of $V_l \setminus \N[B_l]$ or $V_r \setminus \N[B_r]$, and $\s_l(w)$, $\s_r(w)$ cover both cases.

Suppose that $w$ is in $T|\N[B_l]$ but not in $T|\N[B_r]$.
If all leaves in $T^*(w)$ are in $V_l$, having $\s(w) = \s_l(w)$ as in the recurrence is correct.
Suppose that $L(T^*(w)) \setminus V_l$ is non-empty, and choose $u$ among those of minimum distance to $w$.  
Note that $u \notin \N[B_r]$, as otherwise $w$ would be in $T|\N[B_r]$.
Thus $u \in V_r \setminus \N[B_r]$.
We repeat the same type of argument one last time (this is a tedious endeavor, but unfortunately, the cases appear to be not similar enough to handle them with a reusable claim).
Let $v \in L(T(w))$.  We note that $v \in \N[B_l] \setminus B_l = N(B_l) \cap V_l$, since otherwise $w$ would be in $T|\N[B_r]$.  In particular, $uv \notin E(G)$, by the properties of tree decompositions.
Let $d_{u} = dist_{T^*}(u, w)$ and $d_{v} = dist_{T^*}(v, w)$.
Assume that $d_{u} \leq d_{v}$.  Let $v' \in N(v) \cap B_l$, which we know exists.  Note that $v' \in B_r$ as well.  Thus, the path from $v$ to $v'$ in $T^*$ must go through $w$, since otherwise $T|\N[B_r]$ would contain $w$.  Then, $d_{u} \leq d_{v}$ implies that $dist_{T^*}(u, v') \leq k$, a contradiction since $u$ has no neighbor in $B_l$.  Assume instead that $d_{u} > d_{v}$.  By the choice of $u$, there must exist $u' \in V_r \setminus B_r$ such that $uu' \in E(G)$, and such that $u' \notin L(T^*(w))$ (otherwise, $u$ cannot be connected to the rest of the graph).  This implies that $vu' \in E(G)$.  This is not possible by the properties of tree decompositions, since $v \in N(B_l) \subseteq V_l \setminus B_i$, and $u' \in V_r \setminus B_i$.
We deduce that $w$ has no descending leaf in $V_r \setminus \N[B_r]$, and that $\s(w) = \s_l(w)$ is correct.

Suppose that $w$ is in $T|\N[B_r]$ but not in $T|\N[B_l]$.
A symmetric argument justifies that $\s(w) = \s_r(w)$.

The fact that for $u, v \in L(T)$, $uv \in E(G) \Leftrightarrow dist_T(u, v) \leq k$ follows from the fact that $T = T^*|\N[B_i]$ and that $T^*$ is a $k$-leaf root of $G[V_i]$.

Now let $u \in \N[B_l] \setminus B_i$ and take an internal node $w \in V(T|\N[B_r])$.  We note that $u \in V_l \setminus B_i$.  Moreover in $T^*(w)$, there is a leaf $v$ of $V_r \setminus \N[B_r]$ at distance $\s_r(w)$ from $w$.  
It follows that $uv \notin E(G)$, and thus that $dist_{T^*}(u, w) + \s_r(w) > k$.  This implies that $dist_T(u, w) + \s_r(w) > k$. 
A symmetric argument justifies the same property for $u \in \N[B_r] \setminus B_i$.

Finally, let $w_l \in V(T|\N[B_l])$ and $w_r \in V(T|\N[B_r])$. 
Then any leaf of $T^*(w_l) \setminus \N[B_i]$ must be at distance more than $k$ to any leaf of $T^*(w_r) \setminus \N[B_i]$.  
This justifies $\s_l(w_l) + dist_T(w_l, w_r) + \s_r(w_r) > k$.

\medskip 

\noindent
($\Leftarrow$)
Suppose that $(T, \s) \in Q[B_i]$.  We must argue that there is a $k$-leaf root $T^*$ of $G[V_i]$, with $r(T^*) = r(T)$, such that $(T, \s)$ is the valued restriction of $T^*$ to $\N[B_i]$.

\medskip 

\noindent 
\paragraph{Introduce node.}

Let $B_i$ be an introduce node with child $B_j = B_i \setminus \{v\}$.
Then there exists $(T_j, \s_j) \in Q[B_j]$ that satisfies the properties of the recurrence.  Let $\mu$ be the leaf-isomorphism from $T|\N[B_j]$ to $T_j$.
By induction, there is a $k$-leaf power $T^*_j$ of $G[V_j] = G[V_i \setminus \{v\}]$ such that $(T_j, \s_j)$ is the valued restriction of $T^*_j$ to $\N[B_j]$, and such that $r(T^*_j) = r(T_j)$.

We can construct a $k$-leaf root $T^*$ of $G[V_i]$ as follows.
For each $w \in V(T|\N[B_j])$, let $U(w) = ch_{T^*_j}(\mu(w)) \setminus ch_{T_j}(\mu_j(w))$.  
Then for each $u \in U(w)$, add the $T^*_j(u)$ subtree as a child of $w$ in $T$.  Because $r(T^*_j) = r(T_j)$, every leaf in $V_j \setminus \N[B_j]$ gets inserted.  In essence, we simply add the subtrees of $T^*_j$ that $T$ is missing at the locations indicated by $T_j$.

We claim that $T^*$ meets all the requirements for $(T, \s)$ to be valid.  We note that $r(T^*) = r(T)$ since in our construction, we started with $T$, and only added subtrees under some of its internal nodes (thus, we have not accidentally ``raised'' the root in $T^*$).
Let us argue that $(T, \s)$ is the valued restriction of $T^*$ to $\N[B_i]$.
Notice that $T = T^*|\N[B_i]$, since $L(T) = \N[B_i]$, and to obtain $T^*$ we took $T$ and only added extra subtrees under its existing nodes.
Moreover, it is not hard to see that $T^*|V_j \simeq_L T^*_j$.  This is because we started with $T$ such that $T|\N[B_j] \simeq_L T_j$, and we added the missing subtrees of $T^*_j$ on the vertices of $T|\N[B_j]$ at the appropriate locations.

Consider $w \in V(T) \setminus L(T)$ and let us argue that $\s(w)$ represents the correct distance to a hidden leaf under $w$.  
First assume that $w \notin V(T|\N[B_j])$.  Then one can see from the construction that we never insert subtrees under $w$, and hence $\s(w) = \infty$ represents the correct distance.
Assume instead that $w \in V(T|\N[B_j])$.
Then our insertion procedure only adds the leaves descending from children in $ch_{T^*_j}(\mu(w)) \setminus ch_{T_j}(\mu(w))$  under $w$.  The $\s_j(\mu(w))$ quantity gives the minimum  distance from $\mu(w)$ to such an inserted leaf and in this case, the recurrence specifies that $\s(w) = \s_j(\mu(w))$, which is correct.

Let us now argue that $T^*$ is a $k$-leaf root of $G[V_i]$. Because $T^*|V_j \simeq_L T^*_j$, for each $x, y \in V_j$, $xy \in E(G) \Leftrightarrow dist_{T^*}(x, y) \leq k$.  
As for $v$, we know by the recurrence that for each $w \in \N[B_i]$, $vw \in E(G) \Leftrightarrow dist_T(v, w) = dist_{T^*}(v, w) \leq k$.
For any $w \in V_i \setminus \N[B_i]$, let $w'$ be the lowest ancestor of $w$ in $T^*$ such that $w' \in V(T)$, which exists since $r(T) = r(T^*)$.  The recurrence ensures that $dist_{T}(v, w') + \s(w') > k$, and thus that $dist_{T^*}(v, w) = dist_{T}(v, w') + dist_{T^*}(w', w) > k$.  It follows that $T^*$ is a $k$-leaf root of $G[V_i]$.

\medskip 

\noindent 
\paragraph{Forget node.}
Let $B_i$ be a forget node with child $B_j = B_i \cup \{v\}$.
Then there exists $(T_j, \s_j) \in Q[B_j]$ that satisfies the properties of the recurrence.  Let $\mu$ be the leaf-isomorphism from $T$ to $T_j|\N[B_i]$.

By induction, there is a $k$-leaf-power $T^*_j$ of $G[V_j]$
such that $(T_j, \s_j)$ is the valued restriction of $T^*_j$ to $\N[B_j]$, and whose root is $r(T^*_j) = r(T_j)$.
Since $V_j = V_i$, $T^*_j$ is also a $k$-leaf root of $G[V_i]$.
Let us note that since the recurrence requires $r(T_j)$ to have two children with descending leaves in $\N[B_i]$, we have $\mu(r(T)) = r(T_j) = r(T_j^*)$.
We want to argue that $T$ is the valued restriction of $T^*_j$ to $\N[B_i]$.  
Strictly speaking, $V(T)$ is not a subset of $V(T^*_j)$, so we cannot say that $T$ is the valued restriction of $T^*_j$ to $\N[B_i]$ (recall that valued restrictions require usage of same set of nodes).  The simplest solution is to construct a new $k$-leaf root $T^*$ ``around" $T$.

For each $w \in V(T)$, let $C(w) = \{c \in ch_{T_j}(\mu(w)) : \mu^{-1}(w) = \emptyset\}$.  Moreover, let $C^*(w) = \{c \in ch_{T^*_j}(\mu(w)) \setminus ch_{T_j}(\mu(w))\}$.
Then for each $u \in C(w) \cup C^*(w)$, add the $T^*_j(u)$ subtree as a child of $w$ in $T$.  
Because $\mu(r(T)) = r(T_j) = r(T^*_j)$, we know that every leaf of $V_j \setminus \N[B_i]$ is inserted by this procedure.
Moreover, it is not hard to see that $T^* \simeq_L T^*_j$, since we add all the missing subtrees to $T$ at the appropriate locations specified by $T_j$.  Hence, $T^*$ is a $k$-leaf root of $G[V_i]$.
It remains to argue that $(T, \s)$ is the valued restriction of $T^*$ to $\N[B_i]$.

Since $T^*$ is obtained from $T$ by inserting child subtrees under the nodes of $T$, we have $T = T^*|\N[B_i]$ and $r(T) = r(T^*)$.
Now let $w \in V(T) \setminus L(T)$.  Then in $T^*$, consider the minimum distance from $w$ to a leaf in $u \in V_i \setminus \N[B_i]$ that descends from a node in $ch_{T^*}(w) \setminus ch_T(w)$.  
Because $T^* \simeq_L T^*_j$,
this is identical to the $\Rightarrow$ direction of forget nodes. 
That is, the $u$ leaf is either in $L(T_j) \setminus L(T)$, in which case $dist_{T^*}(w) = \min_{l \in \hat{L}} dist_{T_j}(\mu(w), l)$ as in the recurrence, or this leaf is not in $T_j$.
In the latter case, $u$ has some $w'$ as its lowest ancestor in $T_j$, and the distance to $u$ is given by $dist_{T_j}(\mu(w), w') + \s_j(w')$.  The value of $\s(w)$ should be the minimum of all possibilities, as in the recurrence.
Thus $(T, \s)$ is indeed the valued restriction of $T^*$ to $\N[B_i]$.

\medskip 

\noindent 
\paragraph{Join node.}
Let $B_i$ be a join node with children $B_l, B_r$, with $B_i = B_l = B_r$.
Then there exist $(T_l, \s_l) \in Q[B_l]$ and $(T_r, \s_r) \in Q[B_r]$ that satisfy the properties of the recurrence.  Let $\mu_l$ and $\mu_r$ be the leaf-isomorphisms from $T|\N[B_l]$ to $T_l$, and from $T|\N[B_r]$ to $T_r$, respectively.
By induction, there is a $k$-leaf root $T^*_l$ of $G[V_l]$
such that $(T_l, \s_l)$ is the valued restriction of $T^*_l$ to $\N[B_l]$, with $r(T^*_l) = r(T_l)$. 
Likewise, there is a $k$-leaf root $T^*_r$ of $G[V_r]$
such that $(T_r, \s_r)$ is the valued restriction of $T^*_r$ to $\N[B_r]$, with $r(T^*_r) = r(T_r)$.

We can construct $T^*$ as follows.
For each $w \in V(T|\N[B_l])$, let $U_l(w) = ch_{T^*_l}(\mu_l(w)) \setminus ch_{T_l}(\mu_l(w))$.  
Then for each $u \in U_l(w)$, add the $T^*_l(u)$ subtree as a child of $w$ in $T$.  Because $r(T^*_l) = r(T_l)$, every leaf in $V_l \setminus \N[B_l]$ gets inserted.
Similarly, for each $w \in V(T|\N[B_l])$, let $U_r(w) = ch_{T^*_r}(\mu_r(w)) \setminus ch_{T_r}(\mu_r(w))$.  
Then for each $u \in U_r(w)$, add the $T^*_r(u)$ subtree as a child of $w$ in $T$.  Again, every leaf in $V_r \setminus \N[B_r]$ gets inserted.

We claim that $T^*$ meets all the requirements for $(T, \s)$ to be valid.  We note that $r(T^*) = r(T)$ since in our construction, we started with $T$, and only added subtrees under some of its internal nodes.
Let us argue that $(T, \s)$ is the valued restriction of $T^*$ to $\N[B_i]$.
Notice that $T = T^*|\N[B_i]$, since $L(T) = \N[B_i]$, and to obtain $T^*$ we took $T$ and only added extra subtrees under its existing nodes.
Moreover, it is not hard to see that $T^*|V_l \simeq_L T^*_l$.  This is because we started with $T$ such that $T|\N[B_l] \simeq_L T_l$, and we added the missing subtrees of $T^*_l$ on the vertices of $T|\N[B_l]$ at the appropriate locations.
Similarly, $T^*|V_r \simeq_L T^*_r$. 

Consider $w \in V(T) \setminus L(T)$ and let us argue that $\s(w)$ represents the correct distance to a hidden leaf under $w$.  
If $w \in V(T|\N[B_l])$ but is not is $V(T|\N[B_r])$, 
then our insertion procedure only adds the leaves descending from children in $ch_{T^*_l}(\mu_l(w)) \setminus ch_{T_l}(\mu_l(w))$  under $w$.  The $\s_l(\mu_l(w))$ quantity gives the minimum  distance from $\mu_l(w)$ to such an inserted leaf and in this case, the recurrence specifies that $\s(w) = \s_l(\mu_l(w))$, which is correct.
The same argument applies if $w \in V(T|\N[B_r])$ but is not is $V(T|\N[B_l])$.
If $w \in V(T|\N[B_l]) \cap V(T|\N[B_r])$, we have inserted both leaves under $\mu_l(w)$ from $T^*_l$ and leaves under $\mu_r(w)$ from $T^*_r$.  In this case, $\s(w) = \min(\s_l(\mu_l(w)), \s_r(\mu_r(w)))$ correctly represent the minimum distance from $w$ to such a leaf.
It follows that $(T, \s)$ is the valued restriction of $T^*$ to $\N[B_i]$.

Let us now argue that $T^*$ is a $k$-leaf root of $G[V_i]$. Because $T^*|V_l \simeq_L T^*_l$, for each $u, v \in V_l$, $uv \in E(G) \Leftrightarrow dist_{T^*}(u, v) \leq k$.  
Similarly, $T^*|V_r \simeq_L T^*_r$ implies that for each $u, v \in V_r$, $uv \in E(G) \Leftrightarrow dist_{T^*}(u, v) \leq k$.  

Now, consider $u \in \N[B_l]$ and $v \in \N[B_r]$.  The fact that $uv \in E(G) \Leftrightarrow dist_{T^*}(u, v) \leq k$ is explicitly stated in the recurrence, using the fact that $T^*|\N[B_i] = T$.

Next, consider $u \in \N[B_l]$ and $v \in V_r \setminus \N[B_r]$.  
Then $uv \notin E(G)$, by the properties of tree decompositions.  Let $w$ be the lowest ancestor of $v$ in $T^*$ such that $w \in V(T|\N[B_r])$ (which must exist, by the construction of $T^*$, since the subtree containing $v$ was appended under a node of $T|\N[B_r]$ at some point).  In $T^*$, the path from $u$ to $v$ must pass through $w$.
We have $dist_{T^*}(u, v) = dist_T(u, w) + dist_{T^*}(w, v) \geq dist_T(u, w) + \s_r(w)$, which is strictly greater than $k$, by the properties of the recurrence.
The symmetric argument holds for $u \in \N[B_r]$ and $v \in V_l \setminus \N[B_l]$.

Finally, consider $u \in V_l \setminus \N[B_i]$ and $v \in V_r \setminus \N[B_i]$.  Then $uv \notin E(G)$.  Let $w_u$ (resp. $w_v$) be the lowest ancestor of $u$ (resp. $v$) in $T^*$ such that $w_u \in V(T|\N[B_l])$ (resp. $w_v \in V(T|\N[B_r])$).  
Then $dist_{T^*}(u, v) = dist_{T^*}(u, w_u) + dist_{T^*}(w_u, w_v) + dist_{T^*}(w_v, v) \geq \s_l(w_u) + dist_T(w_u, w_v) + \s_r(w_v)$, which is strictly greater than $k$, by the properties of the recurrence.
We have handled every possible pair of leaves and deduce that $T^*$ is indeed a $k$-leaf root of $G[V_i]$.
This concludes the proof.
\end{proof}

We conclude this section with the proof of the main lemma.

\begin{proof}[Proof of Lemma~\ref{lem:enumerate}]
We first need to show that the dynamic programming procedure above can be used to enumerate $\mathbb{T} = \{\T_1, \ldots, \T_l\}$.  
Recall that $\T_i \in \mathbb{T}$ if and only if there is a $k$-leaf root $T^*$ of $G$ such that the root of $T^*$ is the parent of $z$, and $\T_i = (T_i, \s_i)$ is the valued restriction of $T^*$ to $N[z]$.
Let $B_z = \{z\}$ be the root of the tree decomposition from above.

Let $\T_i = (T_i, \s_i) \in \mathbb{T}$ and let $T^*$ be a corresponding $k$-leaf root.  Then $L(T_i) = N[z] = N[B_z] = N[B_z] \cap V_z$ (since $V_z = V(G)$).  Moreover, $(T_i, \s_i)$ is the valued restriction of $T^* = T^*[V_z]$ to $N[z] = N[B_z] \cap V_z$, and $r(T_i) = r(T^*)$ since both roots must be the parent of $z$.
Thus $(T_i, \s_i)$ is valid for $B_z$, and 
it follows from Lemma~\ref{lem:q-is-correct} that $(T_i, \s_i) \in Q[B_z]$.

Now let $(T, \s) \in Q[B_z]$ such that $r(T)$ is the parent of $z$.
Then by Lemma~\ref{lem:q-is-correct}, $(T, \s)$ is valid for $B_z$ and there is a $k$-leaf root $T^*$ of $G[V_z] = G$ such that $(T, \s)$ is the valued restriction of $T^*$ to $N[B_z] \cap V_z = N_G[z]$.  Moreover, $r(T^*) = r(T)$ is the parent of $z$.  
Therefore, $(T, \s)$ must belong to $\mathbb{T}$.

We have thus shown that by enumerating all the valued trees in the computed $Q[B_z]$ and keeping only those whose root is the parent of $z$, we reconstruct exactly $\mathbb{T}$.

Let us discuss the complexity.  Note that $B$ has $O(n)$ nodes.
For each bag $B_i \in V(B)$, by Lemma~\ref{lem:valid-bounded}, we must enumerate $O(f(k))$ possible valued trees and test each of them for membership in $Q[B_i]$, where $f(k) = d^{4kd^{4k}} \cdot (k + 2)^{d^{4k}}$.
Join nodes take the longest to test, since they require to check every combination of valued trees in $Q[B_l]$ and $Q[B_r]$, which amounts to $f(k)^2$ tests.  It is not hard to see that for each combination, checking whether the recurrence holds can be done in time $O(d^{4k} \cdot d^{4k})$ (the longest condition to check is to test each $w_l, w_r$ pairs).
Therefore, each $B_i$ can be computed in time $O(d^{8k} \cdot f(k)^3)$.
Since there are $O(n)$ such $B_i$ bags, the complexity is $O(n \cdot d^{8k} \cdot f(k)^3)$.
\end{proof}

\section{Putting it all together}

The results accumulated above lead to an immediate algorithm.  
First, we check whether $G$ admits a $k$-leaf root of arity at most by $d^k$, where $d = 3|S(k, 3k)|2^{|S(k, 3k)|}$.  This can only happen if $G$ has maximum degree at most $d^k$, in which case we can use the algorithm of Eppstein and Havvaei~\cite{eppstein2020parameterized} (or even our algorithm from Section 4 would work).
If there is no such $k$-leaf root but that $G$ is a $k$-leaf power, by Lemma~\ref{lem:hom-struct}, we must be able to find a homogeneous similar structure $\S$ of size $3|S(k, 3k)|$.

To find $\S$, we begin by searching for $\C$.  We brute-force every $3|S(k, 3k)|$ disjoint subsets of at most $d^k$ vertices from $G$, which is the only reason our algorithm takes time $O(n^{f(k)})$ instead of $O(f(k) n^c)$.
Assuming a suitable $\C$ has been identified, we look at the connected components obtained after removing the $C_i$'s.  At this point, it is easy to verify that the properties of similar structures hold and to find $z$ and the $Y_i$'s.
As for the layering functions, we brute-force them all, but since the size of the $C_i$'s is bounded, this adds little complexity compared to the gargantuan time taken to enumerate the possible $\C$'s.
Once a suitable set of layering functions is found, we use the algorithm from Section 4 to compute all the $accept(\S, C_i)$ sets, and it remains to check that they are equal. 
If so, we have found a redundant substructure of $G$.  By Theorem~\ref{thm:iffc1}, we may remove $C_1 \cup Y_1$ from $G$ to obtain an equivalent instance.  We then repeat with $G - (C_1 \cup Y_1)$.  The algorithm ends when it either reaches a graph of maximum degree at most $d^k$, which is ``easy'' to verify, or when no homogeneous similar structure can be found.  In the latter case, we know by Lemma~\ref{lem:hom-struct} that $G$ cannot be a $k$-leaf power.

\begin{algorithm2e}[h]
\SetAlgoLined
\SetKwProg{Fn}{Function}{}{end}
\Fn{isLeafPower($G, k$)}
{
    $d \gets 3|S(k, 3k)|2^{|S(k, 3k)|}$\;
    \uIf{$G$ has maximum degree at most $d^k$}
    {
        Check if $G$ is a $k$-leaf power and return the result\; \label{line:bounded-degree}
    }
    \ForEach{collection $\C = \{C_1, \ldots, C_l\}$ of disjoint subsets of $V(G)$, with $l = 3|S(k, 3k)|$ and each $|C_i| \leq d^k$}
    {
        Let $G' = G - \bigcup_{i \in [l]} C_i$\;
        Let $X = \{X_1, \ldots, X_t\}$ be the connected components of $G'$\;
        Let $z \in V(G')$ such that $\bigcup_{i \in [l]} C_i \subseteq N_G(z)$\; \label{line:the-z}
        \uIf{$z$ does not exist}
        {
            continue to the next $\C$\;
        }
        Let $X_z \in X$ such that $z \in X_z$\; \label{line:ccs1}
        \uIf{some $X_j \in X \setminus \{X_z\}$ has neighbors in two distinct $C_i, C_j$}
        {
            continue to the next $\C$\;
        }
        For $i \in [l]$, let $Y_i$ be the union of every $X_j \in X \setminus X_z$ such that $N_G(X_j) \subseteq C_i\}$\; \label{line:ccs2}
        \uIf{$\exists i \in [l], G[C_i \cup Y_i \cup \{z\}]$ has maximum degree above $d^k$ \label{line:yi-size}}  
        {
            continue to the next $\C$\;
        }
        \ForEach{set of layering functions $\L = \{\l_1, \ldots, \l_l\}$}
        {
            \uIf{$\S = (\C, \Y = \{Y_1, \ldots, Y_d\}, z, \L)$ is a similar structure}
            {
                \ForEach{$i \in [l]$}
                {
                    Compute $accept(\S, C_i)$\;
                }
                \uIf{all the $accept(\S, C_i)$ are equal and non-empty}
                {
                    return $isLeafPower(G - (C_1 \cup Y_1), k)$ \label{line:homreccall}\;
                }
            }
        }
    }
    return ``Not a $k$-leaf power"\;
}
\caption{Deciding if a graph is a $k$-leaf power.}
\label{alg:findhomsim}
\end{algorithm2e}

\begin{theorem}
Let $k$ be a fixed positive integer.  Then Algorithm~\ref{alg:findhomsim}  correctly decides whether a graph $G$ is a $k$-leaf power, and runs in time $O(n^{(d^k + 1)3|S(k, 3k)| + 6})$, where $d = 3|S(k, 3k)| \cdot 2^{|S(k, 3k)|}$.
\end{theorem}

\begin{proof}
We argue correctness and complexity separately.

\noindent 
\emph{Correctness.}
Assume that $G$ admits a $k$-leaf root of arity at most $d$. Then it will be found on line~\ref{line:bounded-degree}, by Lemma~\ref{lem:prelim:bounded-arity}.  
Otherwise, if $G$ is a $k$-leaf power, all its $k$-leaf roots have arity at least $d + 1$.  By Lemma~\ref{lem:hom-struct}, $G$ admits a homogeneous similar structure $\S = (\C, \Y, z, \L)$ with $|\C| = 3|S(k, 3k)|$, with each $|C_i| \leq d^k$ and $G[C_i \cup Y_i \cup \{z\}]$ having maximum degree $d^k$ or less.

We show that the algorithm finds such a structure, if one exists.
By brute-force, the main for loop will find a $\C$ that belongs to a desired homogeneous similar structure $\S = (\C, \Y, z, \L)$.
The $z$ vertex described on line~\ref{line:the-z} exists, by Property~\ref{cut:znbrhood} of similar structures.
By Property~\ref{cut:ccs}, only the connected component $X_z$ of $G'$ that contains $z$ can have neighbors in more than one $C_i$, and all the others have neighbors in exactly one $C_i$.  Thus lines~\ref{line:ccs1}-\ref{line:ccs2} correctly build the $Y_i$ subsets.  Moreover, by the properties described in Lemma~\ref{lem:hom-struct}, each $G[C_i \cup Y_i \cup \{z\}]$ must have maximum degree at most $d^k$, and thus we will not enter the \emph{if} on line~\ref{line:yi-size}.
After that, since we brute-force every possible set of layering functions, we will eventually find the correct $\L$ for $\S$.  
We then explicitly check whether $\S$ is a similar structure and compute all the $accept$ sets to verify homogeneity.  Assuming that such a $\S$ exists, it follows that it will be found, and that we will eventually reach line~\ref{line:homreccall}.

We can also argue the converse, i.e. that when line~\ref{line:homreccall} is reached, $\S$ is homogeneous and satisfies all the requirements of Lemma~\ref{lem:hom-struct}.  When this line is reached, $\S$ is a similar structure (this is checked explicitly), and is homogeneous since we compute every $accept$ set.  We have $|\C| = 3|S(k, 3k)|$ and each $|C_i| \leq d^k$, since this is what we enumerate.  Moreover, it is checked that each $G[C_i \cup Y_i \cup \{z\}]$ has degree bounded by $d^k$.  Therefore, when line~\ref{line:homreccall} is reached, $\S$ meets all the requirements of Lemma~\ref{lem:hom-struct}.
We can thus apply Theorem~\ref{thm:iffc1} and state that $G$ is a $k$-leaf power if and only if $G - (C_1 \cup Y_1)$ is a $k$-leaf power.  Thus the recursive call on line~\ref{line:homreccall} is correct.

If the algorithm never reaches line~\ref{line:homreccall}, then by the above, $G$ does not admit an homogeneous structure with all the desired properties.  By contraposition of Lemma~\ref{lem:hom-struct}, $G$ cannot be a $k$-leaf power.
This proves the correctness of the algorithm.

\noindent 
\emph{Complexity.}  We can handle the case where $G$ has maximum degree at most $d^k$ in time $O(n (d^k k)^{c d^k})$ for some constant $c$, by Lemma~\ref{lem:prelim:bounded-arity}.
The enumeration of the possible $\C$'s requires choosing $l = 3|S(k, 3k)|$ subsets of $V(G)$ of size at most $d^k$.
We can asymptotically bound the number of $\C$'s to enumerate by 
\[
\left(\sum_{i=1}^{d^k} {n \choose i} \right)^{3|S(k, 3k)|} \leq \left(\sum_{i=1}^{d^k} n^i \right)^{3|S(k, 3k)|} \leq \left( n^{d^k + 1} \right)^{3|S(k, 3k)|}
\]
(for large enough $n$).  For each such $\C$, the construction and verifications for $G', X, z$ and the $Y_i$'s can be done in time $O(n)$, until we must enumerate every set of layering functions.
Such a $\L$ must assign each vertex in $\C$ an integer between $0$ and $k$.  The total number of vertices in $\C$ is at most $3|S(k, 3k)| \cdot d^k$, and so the number of layering functions is at most 
$(k + 1)^{3|S(k, 3k)| \cdot d^k}$.

Then we must check whether $\S$ is truly a similar structure. 
At this point, we must only check that $\L$ satisfies all the requirements of similar structures, which can be done in time $O(n^3)$, since if suffices to compare pairs of vertices of $\C$ and their neighborhoods.

Using Lemma~\ref{lem:enumerate}, one can see that we can compute one $accept(\S, C_i)$
in time $O(|S(k, 3k)| \cdot n^2 d^{8k} f(k)^3)$.
To see this, recall that we can enumerate every valued restriction of $k$-leaf roots for $C_i \cup \{z\} \cup Y_i$ in time $n d^{8k} f(k)^3$.  We need to compute the signature of each of those valued trees.  Such a signature can be computed by traversing each node of the valued tree in post-order.  Each node requires filling a vector with at most $|S(k, 3k)|$ entries, and so computing the signature takes time $O(|S(k, 3k)| \cdot n)$.  Since there are $3|S(k, 3k)|$ $C_i$'s to consider, computing every $accept$ set takes time $O(|S(k, 3k)|^2 \cdot n^2 d^{8k} f(k)^3)$.

Finally, we note that at each recursion, $G$ becomes smaller since $C_1$ is non-empty, so this whole procedure is repeated at most $n$ times.

Let us mention that the complexity of the case of maximum degree at most $d^k$ is dominated by the main loop, so we may omit it.
To sum up we have an asymptotic complexity of
\[ 
n \cdot 
\left(n^{d^k + 1} \right)^{3|S(k, 3k)|} \cdot (k + 1)^{3|S(k, 3k)| \cdot d^k} \cdot n^3 \cdot |S(k, 3k)|^2 \cdot n^2 d^{8k} f(k)^3
\]
where $d, f(k)$, and $|S(k, 3k)|$ depend only on $k$.  
Assuming that $k \in O(1)$, this amounts to 
\[
O(n^{(d^k + 1)3|S(k, 3k)| + 6})
\]

\end{proof}

\section{Conclusion}

Although this work answers a longstanding open question, there is still much to do on the topic of leaf powers and $k$-leaf powers.

\begin{itemize}
    \item 
    Is recognizing $k$-leaf powers fixed-parameter tractable in $k$?  That is, can it be done in time $O(f(k) n^c)$ for some function $f$ and some constant $c$?
    
    Using the techniques of this work would require finding an homogeneous similar structure in FPT time, thereby avoiding brute-force enumeration.  Although this appears difficult, it is possible that such structures have graph-theoretical properties that can be exploited for fast identification.  For instance, we have not used the fact that $G$ is strongly chordal, which may help finding similar structures.
    
    \item 
    Can $k$-leaf powers be recognized in time $O(n^{f(k)})$, where $f(k)$ is more reasonable than in this work?  In particular, can a power tower function be avoided?  It may be possible to find a better type of signature that is more succinct, but still allows proving Theorem~\ref{thm:iffc1}.
    
    
    \item 
    Can the techniques used here be used to recognize leaf powers?  
    In particular, can leaf powers be recognized easily if there is an upper bound on the arity of its leaf root?  Or conversely, is there a structure in leaf powers that admit high arity leaf roots?
    
\end{itemize}

\section*{Acknowledgements}

The author thanks the anonymous reviewers of the SODA 2022 conference for their extremely useful suggestions.

\bibliographystyle{plain}
\bibliography{main}

\end{document}